\documentclass[lettersize,journal]{IEEEtran}
\usepackage{amsmath,amsfonts}
\usepackage{amsthm}
\newtheorem{theorem}{Theorem}
\usepackage{algorithmic}
\usepackage{algorithm}
\usepackage{array}
\usepackage[caption=false,font=normalsize,labelfont=sf,textfont=sf]{subfig}
\usepackage{textcomp}
\usepackage{stfloats}
\usepackage{url}
\usepackage{verbatim}
\usepackage{graphicx}
\usepackage{caption}
\usepackage{cite}
\usepackage{enumitem} 
\begin{document}

\title{A Parallel Region-Adaptive Differential Privacy Framework for Image Pixelization}
\author{Ming Liu%
		\thanks{Ming Liu is with Shanghai Jiao Tong University, Shanghai, China (email: liuming198904@sjtu.edu.cn).}%
	}
\maketitle

\begin{abstract}
The widespread deployment of high-resolution visual sensing systems, coupled with the rise of foundation models, has amplified privacy risks in video-based applications. Differentially private pixelization offers mathematically guaranteed protection for visual data through grid-based noise addition, but challenges remain in preserving task-relevant fidelity, achieving scalability, and enabling efficient real-time deployment. To address this, we propose a novel parallel, region-adaptive pixelization framework that combines the theoretical rigor of differential privacy with practical efficiency. Our method adaptively adjusts grid sizes and noise scales based on regional complexity, leveraging GPU parallelism to achieve significant runtime acceleration compared to the classical baseline. A lightweight storage scheme is introduced by retaining only essential noisy statistics, significantly reducing space overhead. Formal privacy analysis is provided under the Laplace mechanism and parallel composition theorem. Extensive experiments on the PETS, Venice-2, and PPM-100 datasets demonstrate favorable privacy–utility trade-offs and significant runtime/storage reductions. A face re-identification attack experiment on CelebA further confirms the method's effectiveness in preventing identity inference. This validates its suitability for real-time privacy-critical applications such as elderly care, smart home monitoring, driver behavior analysis, and crowd behavior monitoring.
\end{abstract}


\section{Introduction}
\IEEEPARstart{T}{he} rapid deployment of high-resolution video surveillance systems and smart cameras across applications such as public safety, healthcare monitoring, and intelligent transportation has intensified concerns about individual privacy\cite{10.5555/1568647}. These systems often capture unstructured visual data that are directly used in downstream tasks, including feature extraction (e.g., pose estimation\cite{8765346}) and high-level applications (e.g., fall detection\cite{BEDDIAR2022103407} or pedestrian tracking\cite{GUO2025126772}\cite{s24216922}). Meanwhile, the rise of large language models (LLMs) has led to a growing trend where users feed such unstructured data, including images\cite{lee2024llmcxr}\cite{10657904}\cite{10678286}, videos\cite{NEURIPS2024_d7ce06e9}, and audio\cite{hu-etal-2024-wavllm} into powerful models for analysis without proper safeguards. This practice poses a serious risk of leaking sensitive Personally Identifiable Information (PII)\cite{10.1145/3712001} embedded in visual or acoustic content \cite{10.5555/3620237.3620531}.

To address such privacy and information leakage risks, privacy-preserving visual data processing has become a critical research focus. The goal is to anonymize or obfuscate identifiable features in visual data while maintaining the structural and semantic integrity necessary for critical tasks\cite{10657054}\cite{7775034}. Differential privacy (DP)\cite{10.1007/11787006_1}, with its rigorous mathematical guarantees, offers a promising solution for limiting individual information leakage during data sharing or analysis\cite{10.1561/0400000042}. The original differentially private pixelization method\cite{10.1007/978-3-319-95729-6_10} combines pixelization with DP to prevent re-identification. It provides a rigorous mathematical proof of privacy protection, strengthening its theoretical guarantees. It balances privacy and utility by using grid cells and controlled noise, reduces attack success rates, and allows for customizable privacy levels, making it an effective approach for privacy-preserving image and video sharing.

However, applying DP pixelization to visual data remains challenging, particularly in maintaining task-relevant fidelity for segmentation, behavior analysis, or action recognition, which often lacks efficiency and structural adaptability. The CPU-based approach is not scalable for high-resolution data, and the uniform grid-based scheme tends to either over-sanitize or under-protect specific regions. Moreover, the lack of efficient storage and restoration mechanisms further complicates the deployment in real-time or resource-constrained applications. 

Motivated by these limitations, we aim to develop a parallel and region-adaptive pixelization framework that ensures both computational scalability and visual utility. Specifically, our goal is to design an algorithm that (1) exploits GPU parallelism for substantial runtime acceleration, (2) seamlessly differentiates foreground and background regions to achieve better privacy-utility trade-offs, and (3) enables compact storage of noisy statistics for efficient reconstruction and downstream processing. Furthermore, in our experiments, we aim to systematically identify optimal pixelization configurations to maintain utility under differential privacy constraints. Considering these factors, this paper makes the following key contributions:
\begin{itemize}[leftmargin=1em, labelsep=0.5em]
	\item We design an efficient parallel version of differentially private pixelization algorithm, reducing runtime through GPU acceleration and achieving over 30$\times$ speedup over the classical baseline. We thoroughly examine the impact of grid size and padding on processing time.
	
	\item We propose a region-adaptive parallel DP pixelization algorithm that applies varying grid sizes and noise scales based on regional complexity of an image. The non-uniform treatment of complex and simple regions helps preserve essential structures such as human contours. The algorithm also ensures full image coverage, avoiding both unpixelated and excessively pixelated areas.
	
	\item We introduce a lightweight storage and restoration scheme by storing only essential metadata such as noisy grid means, noisy subgrid means, region mask means, grid size, and image size, instead of the full pixelated image. This significantly reduces space complexity compared to full image storage.
	
	\item We provide a formal privacy analysis based on the Laplace mechanism and parallel composition theorem. Experiments on PETS~\cite{5597139}, Venice-2~\cite{MOTChallenge2015}, and PPM-100~\cite{MODNet} demonstrate comparable privacy–utility trade-offs and runtime/storage efficiency. The results show that our method preserves structural fidelity. A face re-identification experiment on CelebA~\cite{liu2015faceattributes} further confirms the method’s effectiveness in mitigating identity inference.
	
\end{itemize}

\section{Related Work}
In zero-trust environments, data custodians may often share unstructured data with third parties for analysis. In practice, such sharing arises in scenarios like cloud-based healthcare monitoring, intelligent transportation analytics, and public safety surveillance, where sensitive visual data must be transmitted. This necessitates rigorous privacy mechanisms beyond access control, motivating the adoption of Differential Privacy (DP). Differential Privacy (DP) provides mathematical guarantees by injecting noise calibrated to a privacy budget $\epsilon$, limiting the influence of any single record. Dwork et al.~\cite{10.1561/0400000042}\cite{10.1007/11681878_14} formalized this via mechanisms such as the Laplace mechanism.

Local Differential Privacy (LDP) protects data in untrusted settings by requiring each user to perturb their data locally, sharing only noisy outputs~\cite{doi:10.1137/090756090}\cite{6686179}. This removes the need for a trusted curator. For example,  RAPPOR~\cite{10.1145/2660267.2660348} achieves large-scale anonymous collection via randomized response. In contrast, centralized methods like P3SGD~\cite{8953573} add noise during training but require server trust. Our work focuses on LDP for large-scale, real-time scenarios where privacy must be ensured at data generation.

Compared to other privacy-preserving methods, local differential privacy offers clear advantages in untrusted environments. Other traditional obfuscation techniques are vulnerable to re-identification via body shape or clothing; for instance, Google's blurring system~\cite{5459413} fails to prevent recognition based on pose. Realistic anonymization~\cite{10209019} lacks formal guarantees, may distort semantic features, and requires costly training. Techniques such as secure multi-party computation~\cite{10474335}, trusted execution environments~\cite{10.1145/3634737.3644993}, and federated learning~\cite{10677989} focus on secure computation but do not protect visual data if devices are compromised. In contrast, DP injects calibrated noise at the source, offering formal, task-independent guarantees and tunable privacy–utility trade-offs.

According to~\cite{10.1145/3490237}, differential privacy has been applied to address the limitations of conventional visual obfuscation. Pixel DP~\cite{10.1007/978-3-319-95729-6_10} adds Laplace noise to image grids, offering strong privacy. Latent Vector DP~\cite{10.1257/pandp.20191109} perturbs feature space selectively to preserve realism. Euclidean Privacy~\cite{8784836}\cite{9578149} adds multivariate noise to SVD features, enabling region-specific protection. Among them, DP pixelization~\cite{10.1007/978-3-319-95729-6_10} provides strong privacy via $m$-neighborhoods, $b \times b$ grid pixelization, and noise addition, achieves fast processing (e.g., 66 ms per image), resists deep re-ID models, and requires no training, making it lightweight and practical compared to CNN- or GAN-based methods~\cite{9578149}\cite{li2021differentiallyprivateimaginglatent}.

\begin{table}[htbp]
	\small
	\caption{Pros and Cons of the Pixel DP Method}
	\label{tab:pixel_dp_pros_cons_en}
	\centering
	\begin{tabular}{|p{2cm}|p{6cm}|}
		\hline
		\textbf{Aspect} & \textbf{Description} \\
		\hline
		\multicolumn{2}{|c|}{\textbf{Pros}} \\
		\hline
		Privacy Strength & Provides the strongest DP protection. Even with high $\epsilon$, it effectively prevents re-identification, making it the most difficult method for attackers to reverse. \\
		\hline
		Simple Implementation & Based on grid partitioning and Laplace mechanism; does not require deep learning models or latent space construction. Suitable for low-resource environments. \\
		\hline
		No Model Dependency & Unlike GAN-based methods \cite{9578149}\cite{li2021differentiallyprivateimaginglatent}, Pixel DP requires no model training, avoiding training costs and hyperparameter tuning. \\
		\hline
		\multicolumn{2}{|c|}{\textbf{Cons}} \\
		\hline
		Image Quality & Pixel-level perturbation significantly degrades visual quality and harms downstream task performance. \\
		\hline
		Perceptual Distortion & Compared to methods operating in high-level semantic space \cite{9578149}\cite{li2021differentiallyprivateimaginglatent}, Pixel DP is less effective at preserving image structure and semantics. \\
		\hline
	\end{tabular}
\end{table}

A summary of the main advantages and limitations of Pixel DP~\cite{10.1007/978-3-319-95729-6_10} is given in Table~\ref{tab:pixel_dp_pros_cons_en}~\cite{10.1145/3490237}. Despite strong privacy guarantees, this low-level method causes semantic loss and low perceptual quality, as reflected by lower SSIM and FID scores~\cite{9578149}\cite{li2021differentiallyprivateimaginglatent}\cite{10.1145/3490237}. Its uniform grid design lacks region adaptivity, yields suboptimal privacy–utility tradeoffs, does not scale well to high-resolution data, and lacks support for efficient storage.

We aim to address the limitations of Pixel DP. The key insight of this paper is that image regions vary in complexity, encoding important structural cues. Complex regions (e.g., faces, bodies, vehicles) contain identity-related features, while simple regions (e.g., sky, walls) are less informative and tolerate coarser processing. This distinction defines visual boundaries and contours, such as human edges, which are critical for scene understanding. To ensure consistent privacy, a uniform differential privacy budget should be applied globally, while protection is adaptively adjusted based on local content. Such region-adaptive mechanism is essential for maintaining formal privacy guarantees while improving utility in downstream visual tasks.

We propose a parallel region-adaptive pixelation framework under differential privacy, where the distinction between complex and simple regions is flexible and customizable for specific applications. In this work, simple background regions are pixelated with coarse grids and Laplace noise, satisfying $\epsilon$-differential privacy while reducing computation time, whereas human regions (identified via segmentation or silhouette extraction) are treated as complex and processed with fine-grained subgrid pixelation. The simple-region grid size must be a multiple of the complex-region subgrid size, ensuring the entire image is covered without leaving unprotected or redundantly processed regions. By tuning parameters ($\epsilon$, $b$, $n$), the framework avoids the extremes of conventional anonymization methods: leaving the background unprotected, which can directly expose sensitive contextual information; or fully masking foreground/background, which can erase attributes such as clothing details, motion dynamics, or environmental cues. Instead, our structure-aware design enables a controllable balance between privacy and utility, benefiting tasks such as elderly care and accident monitoring where contextual and behavioral cues are critical.

We evaluate our parallel DP pixelization algorithm on high-resolution datasets PETS and Venice-2, considering grid size, maximum pixel variation, and privacy budget. Performance is measured using MSE, SSIM~\cite{1284395}, and runtime metrics. Compared to the baseline, our GPU-accelerated version achieves up to 13.3× and 31.5× speedups on PETS and Venice-2, respectively. Further evaluation on PPM-100 confirms that region-adaptive pixelization preserves semantic content while ensuring privacy. The effectiveness of our metadata storage method is validated through experiments on these datasets. Moreover, a face re-identification experiment on CelebA further demonstrates the robustness of our method and its flexibility in privacy–utility customization.

\section{Notations and Symbol Definitions}
To ensure clarity and consistency, we summarize all mathematical symbols used in the algorithms and analysis. Table~\ref{tab:notation_summary} lists the notations and descriptions, categorized by their roles in image structure, grid partitioning, differential privacy, subgrid processing, and related components. This facilitates understanding of the algorithm, theory, and experiments in later sections.

\begin{table}[!t]
	\caption{Notations and Symbol Definitions}
	\label{tab:notation_summary}
	\centering
	\begin{tabular}{|c|p{6.8cm}|}
		\hline
		\textbf{Symbol} & \textbf{Description} \\
		\hline
		\multicolumn{2}{|l|}{Image and Structure} \\
		\hline
		$\mathbf{I}$ & Input grayscale image of size $M \times N$ \\
		$\mathbf{I}(i,j)$ & Pixel value at position $(i,j)$ in $\mathbf{I}$ \\
		$\mathbf{I'}$ & Pixelized image with same size as $\mathbf{I}$ \\
		$M, N$ & Height and width of image $\mathbf{I}$ \\
		$b$ & Grids size for grid partitioning \\
		$m$ & Maximum number of allowed varying pixels (sensitivity parameter) \\
		\hline
		\multicolumn{2}{|l|}{Grid and Subgrid Structures} \\
		\hline
		$G$ & Set of all non-overlapping grids, where each grid \( G_{r,c} \) is of size \( b \times b \), indexed by row \( r \) and column \( c \) \\
		$G_{r,c}$ & Grid at row $r$ and column $c$, size $b \times b$ \\
		$G_R, G_C$ & Number of row-wise and column-wise grids: $G_R = \lceil \frac{M}{b} \rceil$, $G_C = \lceil \frac{N}{b} \rceil$ \\
		$r, c$ & Grid indices in vertical and horizontal directions \\
		$\mathbf{I}_{\text{grid}}$ & Reshaped grid image of size $(G_R, G_C, b, b)$ \\
		\hline
		\multicolumn{2}{|l|}{Pixelization and Mean Values} \\
		\hline
		$\mu_{G_{r,c}}$ & Mean pixel value of grid $G_{r,c}$ \\
		$\mathbf{P}_b(\mathbf{I})$ & Pixelized vector representation of image \\
		$\mu_G$ & Mean of pixel values in each $b \times b$ grid \\
		$\mu'_G$ & Noisy mean after adding $\eta$ to $\mu_G$ \\
		\hline
		\multicolumn{2}{|l|}{Differential Privacy} \\
		\hline
		$\epsilon$ & Privacy budget in $\epsilon$-differential privacy \\
		$\delta$ & Relaxation parameter in $(\epsilon, \delta)$-differential privacy \\
		$\mathcal{M}$ & Randomized mechanism for privacy \\
		$\mathcal{M}_i$ & Mechanism applied to disjoint subset $D_i$ \\
		$D, D'$ & Neighboring datasets differing by one record \\
		$D_i$ & Disjoint subset of dataset $D$ \\
		$k$ & Number of disjoint subsets $D_i$ \\
		$S$ & Subset of possible outputs of $\mathcal{M}$ \\
		$f$ & Function to be privatized \\
		$\Delta f$ & Sensitivity of function $f$ \\
		\hline
		\multicolumn{2}{|l|}{Noise and Sensitivity} \\
		\hline
		$\eta$ & Laplace noise sampled as $\text{Laplace}(0, \sigma)$ \\
		$\eta_{r,c}$ & Laplace noise added to grid cell $(r,c)$ \\
		$\sigma$ & Noise scale\\
		$\Delta$ & Sensitivity \\
		$\Delta \mathbf{P}_b$ & Global sensitivity of $\mathbf{P}_b(\mathbf{I})$ \\
		$\tilde{\mathbf{P}}_b(\mathbf{I})$ & Noisy pixelized vector of the image \\
		$\mathbf{I}_1, \mathbf{I}_2$ & Neighboring images differing in at most $m$ pixels \\
		\hline
		\multicolumn{2}{|l|}{Subgrid and Region Complexity} \\
		\hline
		$\mathbf{M}$ & Binary mask indicating simple vs. complex regions \\
		$\mathbf{M}_{\text{padded}}$ & Padded version of $\mathbf{M}$ \\
		$\mathbf{M}_{\text{grid}}$ & Reshaped mask grid of size $(G_R, G_C, b, b)$ \\
		$\mu_{\text{mask}}$ & Mean of each grid in $\mathbf{M}_{\text{grid}}$ \\
		$\mu_{\text{simple}}$ & Boolean mask indicating simple regions ($>0.5$) \\
		$\mu_{\text{complex}}$ & Boolean mask indicating complex regions ($\leq 0.5$) \\
		$n$ & Subgrid division factor in complex regions ($b = n \cdot s_b$) \\
		$s_b$ & Subgrid size ($s_b = b/n$) \\
		$G_s$ & Subgrids extracted from complex regions \\
		$\mu_{G_s}$ & Mean of subgrids $G_s$ \\
		$\eta_s$ & Laplace noise for $G_s$ \\
		$\mu'_{G_s}$ & Noisy mean after adding $\eta_s$ to $G_s$ \\
		$\Delta_s$ & Sensitivity for $s_b \times s_b$ subgrid \\
		$\sigma_s$ & Noise scale for subgrid \\
		
		\hline
		\multicolumn{2}{|l|}{Padding} \\
		\hline
		$P_R, P_C$ & Number of pixels added for padding rows and columns \\
		$\mathbf{I}_{\text{padded}}$ & Image after padding to $(G_R \cdot b, G_C \cdot b)$ \\
		\hline
	\end{tabular}
\end{table}

\section{The Original Differentially Private Pixelization}
\subsection{Image pixelization}
Consider a gray scale input image $\mathbf{I}$ defined as a two-dimensional matrix \(\mathbf{I} \in \mathbb{R}^{M \times N}\), where each element \(\mathbf{I}(i, j)\) represents the pixel value at position \((i, j)\) in the image, with \(i \in \{1, 2, \dots, M\}\) and \(j \in \{1, 2, \dots, N\}\).We assume the image is grayscale, so \(\mathbf{I}(i, j)\) is a scalar value in the range \([0, 255]\). 
The pixelization process can be described as partitioning the original image $\mathbf{I}$ into non-overlapping grids in the spatial domain. Specifically, the image $\mathbf{I}$ is divided into several grids \(G\) based on a grid size parameter \(b\), where each grid $G_{r,c}$ has dimensions \(b \times b\) and satisfies:
\begin{equation}
	\begin{aligned}
		G_{r,c} = \{\mathbf{I}(i, j) \mid & \, i \in [r \cdot b, (r+1) \cdot b - 1], \\
		& \, j \in [c \cdot b, (c+1) \cdot b - 1]\}
	\end{aligned}
\end{equation}
where $G_{r,c}$ denote the $r,c$-th grid cell of an image $\mathbf{I}$ with dimensions $M \times N$, and \(r \in \{0, 1, \dots, G_R - 1\}\) and \(c \in \{0, 1, \dots, G_C - 1\}\) are the indices of the grid in the image, and \(G_R = \left\lceil \frac{M}{b} \right\rceil\) and \(G_C = \left\lceil \frac{N}{b} \right\rceil\) are the grid dimensions. Therefore, the image is divided into $\left\lceil \frac{M}{b} \right\rceil \left\lceil \frac{N}{b} \right\rceil$ grid cells.
For each grid \(G_{r,c}\), a new pixel value is generated by computing the mean of all pixel values within the grid:
\begin{equation}
	\mu_{G_{r,c}} = \frac{1}{b^2} \sum_{i=r \cdot b}^{(r+1) \cdot b - 1} \sum_{j=c \cdot b}^{(c+1) \cdot b - 1} I(i, j)
\end{equation}
where \(\mu_{G_{r,c}}\) is the mean pixel value of the grid \(G_{r,c}\), which is used as the uniform value for all pixels within that grid.
Next, a new pixelized image $\mathbf{I'}$ is generated, which retains the same dimensions as the original image. However, each pixel value is substituted with the corresponding grid mean \(\mu_{G_{r,c}}\):
\begin{equation}
	\mathbf{I'}(i, j) = \mu_{G_{r,c}},
\end{equation}
\begin{equation*}
	\text{where } i \in [r \cdot b, (r+1) \cdot b - 1], \quad j \in [c \cdot b, (c+1) \cdot b - 1].
\end{equation*}
The pixelization of the image $\mathbf{I}$ can then be represented as a vector $\mathbf{P}_b(\mathbf{I})$ of length $G_R \times G_C$:
\begin{equation}
	\begin{aligned}
		\mathbf{P}_b(\mathbf{I}) = \left\{ \frac{1}{b^2} \sum_{(i,j) \in G_{r,c}} \mathbf{I}(i,j) \mid r = 0, 1, \dots, G_R - 1, \right. \\
		\left. c = 0, 1, \dots, G_C - 1 \right\}.
	\end{aligned}
\end{equation}

The resulting pixelized image $\mathbf{I'}$ is a simplified version of the original image $\mathbf{I}$, where spatial resolution is reduced, and image details are blurred. The original detailed information is replaced by larger uniform grids of mean values, visually appearing as larger, more uniform grids. Through these steps, the pixelization technique simplifies the visual details of the image while providing a degree of privacy protection for potentially sensitive information contained within the image.

\subsection{Differential Privacy and the Laplace Mechanism}

Differential privacy (DP) provides a rigorous mathematical framework to quantify privacy guarantees when analyzing data. Specifically, an mechanism $\mathcal{M}$ is said to satisfy $(\epsilon, \delta)$-differential privacy if, for any two neighboring datasets $D$ and $D'$, which differ by at most one element, and for all subsets $S$ of the mechanism's output space, the following holds:
\begin{equation}
	\mathbb{P}[\mathcal{M}(D) \in S] \leq e^{\epsilon} \times \mathbb{P}[\mathcal{M}(D') \in S] + \delta
\end{equation}
In this paper, we focus on the pure $\epsilon$-differential privacy (where $\delta = 0$), and thus the inequality becomes:
\begin{equation}
	\mathbb{P}[\mathcal{M}(D) \in S] \leq e^{\epsilon} \times \mathbb{P}[\mathcal{M}(D') \in S]
\end{equation}
Here, $\epsilon$ is the privacy budget, controlling the level of privacy protection. A smaller $\epsilon$ implies stronger privacy.

The Laplace mechanism is one of the most commonly used methods to achieve differential privacy. For a given function $f$ with sensitivity $\Delta f$, which is defined as:
\begin{equation}
	\Delta f = \max_{D, D'} ||f(D) - f(D')||
\end{equation}
The Laplace mechanism adds noise drawn from a Laplace distribution to the output of $f$. The noise is parameterized by a scale $\sigma$, which is directly related to the privacy budget $\epsilon$:
\begin{equation}
	\sigma = \frac{\Delta f}{\epsilon}
\end{equation}
The noise $\eta$ added to the output is drawn from the Laplace distribution:
\begin{equation}
	\eta \sim \text{Laplace}(0, \sigma)
\end{equation}
This ensures that the perturbed output satisfies $\epsilon$-differential privacy.

\subsection{Parallel Composition Theorem}
The Parallel Composition Theorem \cite{10.1145/1559845.1559850} ensures that when applying differentially private mechanisms to disjoint subsets of the data, the overall privacy guarantee remains bounded by the worst-case guarantee, not the sum of the individual guarantees.

\begin{theorem}\label{thm:parallel_composition}
Let $\mathcal{M}_i$ be a differentially private mechanism providing $\epsilon$-differential privacy. Let $D_i$ represent an arbitrary disjoint subset of the input domain $D$, such that:
\begin{equation}
	D = \bigcup_{i=1}^{k} D_i, \quad D_i \cap D_j = \emptyset \text{ for } i \neq j.
\end{equation}
where $i$ and $j$ are indices representing different subsets $D_i$ and $D_j$ of the input data set $D$, 
$k$ is the total number of disjoint subsets into which the data set $D$ is partitioned.
The sequence of mechanisms $\mathcal{M}_i(X \cap D_i)$ applied to each disjoint subset $D_i$ ensures that the entire system satisfies $\epsilon$-differential privacy. 

\end{theorem}

The Parallel Composition Theorem states that if a dataset $D$ is divided into disjoint subsets $D_1, D_2, \dots, D_k$, and each subset is analyzed using a differentially private mechanism with privacy budget $\epsilon$, the overall system maintains $\epsilon$-differential privacy. The privacy loss does not accumulate across the disjoint subsets but is instead controlled by the privacy guarantee of the individual mechanisms applied to each subset.

\subsection{The Original Differentially Private Pixelization}
The original Pixelization with Differential Privacy\cite{10.1007/978-3-319-95729-6_10} first pixelizes an input image and then applies Laplace noise to each grid cell to guarantee privacy. 
The global sensitivity $\Delta \mathbf{P}_b$ of the pixelization function $\mathbf{P}_b(\mathbf{I})$ is defined as:
\begin{equation}
	\Delta \mathbf{P}_b = \max_{\mathbf{I}_1, \mathbf{I}_2} \|\mathbf{P}_b(\mathbf{I}_1) - \mathbf{P}_b(\mathbf{I}_2)\| = \frac{255m}{b^2}
\end{equation}
where $\mathbf{I}_1$ and $\mathbf{I}_2$ are any neighboring images differing in at most $m$ pixels. Since each pixel difference is bounded by 255, this defines the sensitivity. To ensure differential privacy, Laplacian noise is added to the pixelized vector:
\begin{equation}
	\tilde{\mathbf{P}}_b(\mathbf{I}) = \mathbf{P}_b(\mathbf{I}) + \eta
\end{equation}
The noise $\eta = \{\eta_{r,c} \mid r = 0, 1, \dots, G_R - 1, \, c = 0, 1, \dots, G_C - 1\}$ is drawn from a Laplace distribution with mean 0 and scale $\sigma = \frac{255m}{b^2\epsilon}$, where $\epsilon$ is the privacy budget.
The perturbed pixelization vector $\tilde{\mathbf{P}}_b(\mathbf{I})$ is then clipped to the range $[0, 255]$ to ensure valid pixel values. According to the definition of differential privacy, the mechanism $\tilde{\mathbf{P}}_b$ guarantees $\epsilon$-differential privacy, as described by Theorem \ref{thm:algorithm_privacy}.

\begin{theorem}\label{thm:algorithm_privacy}
	Mechanism $\tilde{\mathbf{P}}_b$ guarantees $\epsilon$-differential privacy.
\end{theorem}
\begin{proof}
	The sensitivity of the pixelization process is $\Delta \mathbf{P}_b = \frac{255m}{b^2}$, the noise $\eta$ added to each element of the pixelization vector is drawn from a Laplace distribution with scale $\sigma = \frac{\Delta \mathbf{P}_b}{\epsilon} = \frac{255m}{b^2\epsilon}$, ensuring that the probability of the mechanism's output does not differ by more than a factor of $e^\epsilon$ for any neighboring images $\mathbf{I}_1$ and $\mathbf{I}_2$. Therefore, the mechanism $\tilde{\mathbf{P}}_b$ satisfies $\epsilon$-differential privacy.

\end{proof}

\begin{algorithm}[t]
	\caption{The Original Differentially Private Image Pixelization \cite{10.1007/978-3-319-95729-6_10}}
	\label{alg:original_dp_pixelization} 
	\begin{algorithmic}[1]
		\REQUIRE Image $\mathbf{I}$ of size $M \times N$, grid size $b$, privacy budget $\epsilon$, max variation $m$
		\ENSURE Restored Image $\mathbf{I'}$
		\STATE Initialize $\mathbf{I'} \gets \mathbf{0}_{M \times N}$
		\STATE Calculate sensitivity $\Delta \gets \frac{255m}{b^2}$
		\STATE Calculate noise scale $\sigma \gets \frac{\Delta}{\epsilon}$
		\FOR{$r = 0$ to $M-1$ by $b$}
		\FOR{$c = 0$ to $N-1$ by $b$}
		\STATE $h \gets \min(b, M - r)$
		\STATE $w \gets \min(b, N - c)$
		\STATE Extract grid $G_{r,c} \gets \mathbf{I}[r:r+h, c:c+w]$
		\STATE Calculate mean of grid $\mu_{G_{r,c}} \gets \text{mean}(G_{r,c})$
		\STATE Calculate noise $\eta_{r,c} \leftarrow \text{Laplace}(0, \sigma)$
		\STATE Generate noisy mean $\mu_{G_{r,c}}^\prime \gets \mu_{G_{r,c}} + \eta_{r,c}$
		\STATE Clip noisy value $\mu_{G_{r,c}}^\prime \gets \text{clip}(\mu_{G_{r,c}}^\prime, 0, 255)$
		\STATE Fill restored image $\mathbf{I'}[r:r+h, c:c+w] \gets \mu_{G_{r,c}}^\prime$
		\ENDFOR
		\ENDFOR
		\RETURN $\mathbf{I'}$
	\end{algorithmic}
\end{algorithm}
As shown in Algorithm~\ref{alg:original_dp_pixelization}, the original differentially private pixelization algorithm processes an input image $\mathbf{I}$ of size $M \times N$ by dividing it into $b \times b$ grids, achieving pixelation with differential privacy. It initializes an empty matrix $\mathbf{I'}$ of the same size to store the output. The sensitivity $\Delta = \frac{255m}{b^2}$ is calculated, and the noise scale $\sigma = \frac{\Delta}{\epsilon}$ is determined. The algorithm iterates over the image, extracting each grid $G_{r,c}$, computing its mean $\mu_{G_{r,c}}$, and generating a noisy mean $\mu_{G_{r,c}}^\prime = \mu_{G_{r,c}} + \text{Laplace}(0, \sigma)$. This noisy mean is clipped to the range $[0, 255]$ and used to fill the corresponding grid in $\mathbf{I'}$. The resulting image $\mathbf{I'}$ is a pixelated version of $\mathbf{I}$ with added noise to ensure differential privacy.

\subsection{Time and Space Complexity Analysis of Algorithm~\ref{alg:original_dp_pixelization}}
Algorithm~\ref{alg:original_dp_pixelization} follows a grid-wise traversal approach for processing an image of size $M \times N$ with a grid size of $b \times b$. 
Step 1 initializes the output image, which takes $O(1)$ time.
Steps 2 and 3 compute the sensitivity and noise scale using basic arithmetic operations, both having a complexity of $O(1)$.
The algorithm iterates over the image in a nested loop structure. The outer loop (Step 4) iterates over the rows with a step size of $b$, leading to $O(M/b)$ iterations. The inner loop (Step 5) iterates over the columns with a step size of $b$, leading to $O(N/b)$ iterations. Thus, the total number of iterations is $O\left(MN / b^2\right)$. For each grid of size $b \times b$, 
Extracting the grid (Step 8) takes $O(b^2)$ time.
Computing the mean value of the grid (Step 9) requires summing $b^2$ elements and dividing by $b^2$, which results in $O(b^2)$ complexity.
Generating Laplacian noise (Step 10) is an $O(1)$ operation.
Adding noise and clipping the result (Steps 11-12) are both $O(1)$ operations.
Assigning the noisy mean value to all pixels in the grid (Step 13) takes $O(b^2)$ time. Therefore, the total complexity per grid is $O(b^2) + O(b^2) + O(1) + O(1) + O(1) + O(b^2) = O(b^2).$
Since there are $O(MN/b^2)$ such grids, the total complexity of Algorithm~\ref{alg:original_dp_pixelization} is:
\begin{equation}
	O(b^2) \times O\left(MN / b^2\right) = O(MN).
\end{equation}
This is determined by image size. However, grid size $b$ strongly affects sensitivity, noise level, and efficiency: $\Delta \propto 1/b^2$, so larger $b$ reduces noise and iteration count, while smaller $b$ preserves detail but increases cost. Due to its inherently sequential nature, the grid-wise loop impedes parallel execution and results in suboptimal utilization of hardware resources. Vectorized GPU-based parallelization can mitigate the inefficiencies caused by sequential processing. 

The space complexity of Algorithm~\ref{alg:original_dp_pixelization} is $O(M  N)$, dominated by storing $\mathbf{I}$ and $\mathbf{I'}$. The grid $G_{r,c}$ uses at most $O(b^2)$ space. Scalars like $\mu_{G_{r,c}}$, $\eta_{r,c}$, and $\mu_{G_{r,c}}^\prime$ require $O(1)$ space. Thus, the total space complexity is $O(M  N)$.

\section{Parallel Differentially Private Image Pixelization}
This section presents a parallel version of the differentially private image pixelization algorithm. It operates independently across processing units while preserving $\epsilon$-differential privacy via Laplace noise. GPU acceleration using CuPy enables fast and scalable processing for large-scale applications.

\subsection{Parallelizability of Algorithm~\ref{alg:original_dp_pixelization}}
\begin{theorem}
	Algorithm~\ref{alg:original_dp_pixelization} is parallelizable across multiple processing units.
\end{theorem}

\begin{proof}
	Let $\mathbf{I}$ be an image of size $M \times N$, divided into $G_R \times G_C$ non-overlapping grids of size $b \times b$, where $G_R = \left\lceil \frac{M}{b} \right\rceil$ and $G_C = \left\lceil \frac{N}{b} \right\rceil$. Each grid $G_{r,c}$ is processed independently by computing its mean, adding Laplace noise $\eta_{r,c}$, and clipping the noisy mean:
	\begin{equation}
		\mu_{G_{r,c}}' = \text{clip}\left(\frac{1}{b^2} \sum_{(i,j) \in G_{r,c}} \mathbf{I}(i,j) + \eta_{r,c}, 0, 255\right)
	\end{equation}
Since each grid \( G_{r,c} \) is processed using only its own pixels, their computations are independent and can be parallelized across multiple processing units without inter-unit communication. Given $P$ processing units, the $G_R \times G_C$ grids can be evenly distributed, with each unit handling approximately $ \frac{G_R \times G_C}{P} $ grids.
\end{proof}

\subsection{Algorithm Description}
Algorithm~\ref{alg:parallel_dp_image_pixelization} processes an input image $\mathbf{I}$ of size $M \times N$ by dividing it into $b \times b$ grids, achieving pixelation with differential privacy. The sensitivity $\Delta = \frac{255 \times m}{b^2}$ is calculated, and the noise scale $\sigma = \frac{\Delta}{\epsilon}$ is determined. The image is padded to $\mathbf{I}_{\text{padded}}$ of size $(G_R \times b,G_C \times b)$ by mirroring the border pixels, and then reshaped into $\mathbf{I}_{\text{grid}}$ of shape $(G_R, b, G_C, b)$. The algorithm computes the means $\mu_G$ for each grid, and generates noisy means $\mu_G^\prime = \mu_G + \text{Laplace}(0, \sigma)$ for each grid. The $\mu_G^\prime$ is clipped to the range $[0, 255]$ and repeated to restore to the original image size, with the final restored and clipped image $\mathbf{I'}$ being a pixelated version of $\mathbf{I}$ with added noise to ensure differential privacy. 
\begin{theorem}
	Algorithm~\ref{alg:parallel_dp_image_pixelization} preserves $\epsilon$-differential privacy.
\end{theorem}
\begin{proof}
	The image $\mathbf{I}$ is divided into $G_R \times G_C$ disjoint grids, defined as:
	\begin{equation}
		G = \{G_{r,c} \mid r = 1, \dots, G_R,\; c = 1, \dots, G_C\}
	\end{equation}
	where each $G_{r,c} \subseteq \mathbf{I}$ and $G_{r,c} \cap G_{r',c'} = \emptyset$ for $r \neq r'$ or $c \neq c'$. Since each grid is independent, the Laplace mechanism is applied separately to each $G_{r,c}$. The grids are distributed across $P$ processing units, each handling a subset $\mathcal{S}_p \subseteq G$. By Theorem~\ref{thm:parallel_composition}, applying differentially private mechanisms to disjoint subsets yields a total privacy loss bounded by $\epsilon$, not $P \times \epsilon$. Thus, the overall guarantee remains $\epsilon$-differential privacy.
\end{proof}

To ensure $\mu_G$ reflects only original pixels, we use mirror padding to minimize boundary artifacts. Unlike zero or constant padding, which introduces sharp discontinuities, mirror padding preserves edge continuity and maintains the quality of pixelized images, preventing degradation in downstream tasks such as object detection or segmentation.

The algorithm is parallelized using GPU acceleration with CuPy for efficient processing. Key steps include: (1) Vectorizing padding, reshaping, and mean calculation operations for parallel processing across GPU cores, (2) Simultaneously generating and adding noise for all grids, (3) Restoring the image by repeating the noisy means over the grids, also in parallel. These steps distribute the computational load across GPU resources, significantly improving performance compared to serial processing.

\begin{algorithm}[t]
	\caption{Parallel Differentially Private Image Pixelization}
	\label{alg:parallel_dp_image_pixelization} 
	\begin{algorithmic}[1]
		\REQUIRE Image $\mathbf{I}$ of size $M \times N$, grid size $b$, privacy budget $\epsilon$, max variation $m$
		\ENSURE Restored image $\mathbf{I'}$
		\STATE Calculate sensitivity $\Delta \leftarrow \frac{255 \times m}{b^2}$
		\STATE Calculate noise scale $\sigma \leftarrow \frac{\Delta}{\epsilon}$
		\STATE Calculate grid dimensions $G_R \leftarrow \left\lceil \frac{M}{b} \right\rceil$, $G_C \leftarrow \left\lceil \frac{N}{b} \right\rceil$
		\STATE Pad the image to size $(G_R \times b,G_C \times b)$ by mirroring the border pixels
		\STATE Reshape the padded image $\mathbf{I}_{\text{padded}}$ to \\
		$\mathbf{I}_{\text{grid}} \leftarrow (G_R, b, G_C, b)$
		
		\STATE Compute the means of each grid $\mu_G \leftarrow \text{mean}(\mathbf{I}_{\text{grid}}, (1,3))$
		\STATE Generate Laplacian noise $\eta \leftarrow \text{Laplace}(0, \sigma, (G_R, G_C))$
		\STATE Add noise to grid means $\mu'_G \leftarrow \mu_G + \eta$
		\STATE Clip the noisy means $\mu'_G \leftarrow \text{clip}(\mu'_G, 0, 255)$
		\STATE Restore image by repeating $\mu'_G$ over grids $b \times b$
		\STATE Crop the restored image $\mathbf{I'}$ to original size $M \times N$
		\RETURN $\mathbf{I'}$
	\end{algorithmic}
\end{algorithm}

\subsection{Time and Space Complexity of Algorithm~\ref{alg:parallel_dp_image_pixelization}}
\subsubsection{Time complexity}
Step 1–3 (sensitivity, noise scale, grid size) $\rightarrow$ $O(1)$;  
Step 4 (padding) $\rightarrow$ $O(MN)$;  
Step 5 (reshape) $\rightarrow$ $O(1)$;  
Step 6 (grid means) $\rightarrow$ $O(MN)$;  
Step 7–9 (noise generation, addition, clipping) $\rightarrow$ $O(MN / b^2)$ each;  
Step 10 (repeat to restore image) $\rightarrow$ $O(MN)$;  
Step 11 (crop) $\rightarrow$ $O(MN)$.  
Thus, the overall time complexity is $O(MN)$.
Compared to Algorithm~\ref{alg:original_dp_pixelization}, Algorithm~\ref{alg:parallel_dp_image_pixelization} reduces loop nesting through vectorized and batched operations, achieving faster execution via GPU parallelism, despite having the same theoretical time complexity.
\subsubsection{Space complexity}
The space complexity of the algorithm is dominated by the input image $\mathbf{I}$ and the output image $\mathbf{I'}$, both requiring $O(M  N)$ space. The padded image and reshaped grids also contribute $O(M  N)$ space, as the padding and reshaping operations do not significantly alter the overall size. The grid means $\mu_G$ and noise $\eta$ occupy $O(MN / b^2)$ space, which is relatively small compared to the storage for $\mathbf{I}$ and $\mathbf{I'}$. Therefore, the overall space complexity remains $O(M  N)$, but the efficient GPU memory access making it more suitable for large-scale image processing.

\subsection{Efficient Storage of the Pixelated Image by structured representation}
\label{Efficient_Storage_Alg_2}
For large-scale image processing scenarios, instead of the entire pixelated image $\mathbf{I'}$, we propose storing only $\mu'_G$, grid size $b$, and the image size ($M \times N$), as summarized in Table~\ref{tab:stored_quantities_1}. This can drastically compress the storage space. 
\begin{table}[ht]
	\centering
	\caption{Essential Stored Metadata in Algorithm~\ref{alg:parallel_dp_image_pixelization} for pixelated Image Reconstruction}
	\begin{tabular}{|c|p{3.2cm}|p{3.2cm}|}
		\hline
		\textbf{Item} & \textbf{Description} & \textbf{Storage Space} \\
		\hline
		$\mu'_G$ & Noisy means of each $b \times b$ grid& proportional to the number of grids, which is $O(MN / b^2)$ \\
		\hline
		$b$ & Grid size& Constant, i.e., $O(1)$ \\
		\hline
		$M \times N$ & Original image size, required for padding and cropping & Constant, i.e., $O(1)$ \\
		\hline
	\end{tabular}
	\label{tab:stored_quantities_1}
\end{table}

Each noisy mean is a floating-point value, so the space complexity for storing the noisy means is $O(G_R \times G_C) = O(MN / b^2)$. Storing the original image dimensions $M \times N$ and grid size $b$ requires constant space, i.e., $O(1)$. Thus, the total space complexity for storing these metadata is $O(MN / b^2) + O(1) = O(MN / b^2)$. Compared to storing the full image, which requires $O(MN)$ space, this method greatly reduces the space complexity. This is especially beneficial for large-scale image processing.

\section{Region-Adaptive Parallel Differentially Private Image Pixelization}
\subsection{Adaptive Grid Sizes}
In standard differentially private pixelization, the image $\mathbf{I}$ is uniformly divided into $b \times b$ grids. However, this uniform treatment may be suboptimal for regions with varying complexity. We address this by adaptively adjusting grid size $b$ based on local image characteristics. Specifically, the image $\mathbf{I}$ is segmented into complex regions (e.g., regions of interest) and simple regions (e.g., backgrounds) prior to pixelization. 
A predefined mask $\mathbf{M}$ labels each region type. This mask is obtained using edge-detection or segmentation models (e.g., Sobel operator or pretrained human-segmentation networks) to highlight complex regions characterized by high variance.
Simple regions use larger grids $b$, with sensitivity $\Delta = \frac{255 \times m}{b^2}$ and noise scale $\sigma = \frac{\Delta}{\epsilon}$. Complex regions use smaller subgrids $s_b$, with sensitivity $\Delta_s = \frac{255 \times m}{s_b^2}$ and noise scale $\sigma_s = \frac{\Delta_s}{\epsilon}$, ensuring consistent $\epsilon$-differential privacy. Each grid’s mean is computed, noise is added, and value is clipped to [0, 255]. The results of all grids are reshaped to form the final image $\mathbf{I'}$. By adjusting grid sizes, this adaptive approach applies different pixelization granularities to complex and simple regions, preserving important structural features and maintaining utility in each region, thereby balancing privacy, utility, and efficiency.

\begin{theorem}
	Let $\mathbf{I}$ be an image of size $M \times N$, divided into $b \times b$ grids and partitioned into simple and complex regions by a mask $\mathbf{M}$. For simple regions, the sensitivity is $\Delta = \frac{255 \times m}{b^2}$; for complex regions, grids are further divided into subgrids of size $s_b$, with sensitivity $\Delta_s = \frac{255 \times m}{s_b^2}$. The entire process satisfies $\epsilon$-differential privacy.
\end{theorem}

\begin{proof}
	Simple and complex regions are disjoint. Laplace noise with scale $\sigma = \frac{255 \times m}{b^2 \epsilon}$ is added to simple grids, and $\sigma_s = \frac{255 \times m}{s_b^2 \epsilon}$ to complex subgrids. As noise is applied independently to disjoint subsets, Theorem~\ref{thm:parallel_composition} ensures the overall process satisfies $\epsilon$-differential privacy.
\end{proof}

\subsection{Interpreting Privacy Strength Under Varying Grid Sizes}
To maintain consistent $\epsilon$-differential privacy across the image, we adopt a uniform privacy setting where both simple and complex regions share the same privacy budget $\epsilon$ and maximum variation $m$. Under this setting, the overall privacy guarantee depends only on $\epsilon$ and $m$; the choice of grid size $b$ and region location does not affect formal privacy. To enforce this, noise must be scaled according to local sensitivity: $\sigma = \Delta / \epsilon$. For simple regions, $\Delta = \frac{255 \times m}{b^2}$; for complex regions subdivided into $s_b \times s_b$ grids, $\Delta_s = \frac{255 \times m}{s_b^2}$. Smaller grids yield higher sensitivity, requiring larger noise to preserve consistent privacy. Using the same noise scale $\sigma_{\text{same}}$ would result in different privacy levels: $\epsilon_{\text{complex}} = \frac{\Delta_s}{\sigma_{\text{same}}} > \epsilon_{\text{simple}} = \frac{\Delta}{\sigma_{\text{same}}}$, thus violating uniform $\epsilon$-differential privacy.

\subsection{Grid Size Alignment for Seamless Partitioning}

To partition the entire image \( \mathbf{I} \) seamlessly using different grid sizes for simple and complex regions, the grid size for simple regions \( b \) must be an integer multiple of the grid size for complex regions \( s_b \):
\begin{equation}
	b = n \times s_b,\quad 1 \leq n \leq b, \quad n \in \mathbb{Z}^+
\end{equation}
Consequently, the sensitivity in simple regions is proportional to that in complex regions:
\begin{equation}
	\Delta = \frac{255 \times m}{b^2} = \frac{255 \times m}{(n \times s_b)^2} = \frac{\Delta_s}{n^2}
\end{equation}

This formulation ensures consistent sensitivity and noise scaling across regions. By aligning grid sizes, the image can be uniformly divided without partial grids, avoiding complications in pixelization and privacy enforcement. The entire image is fully covered by partitioning it into simple and complex regions without gaps or redundancy. Each $b \times b$ grid is assigned to one type: simple regions (e.g., backgrounds) are directly averaged with noise added, while complex regions (e.g., faces) are further subdivided into smaller subgrids ($s_b = \frac{b}{n}$) for fine-grained processing. 

\subsection{Parallel Differentially Private Image Pixelization with Subgrid Processing}
\begin{algorithm}[t]
	\caption{Parallel Differentially Private Image Pixelization with Subgrid Processing}
	\label{alg:dp_pixelization_subgrid}
	\begin{algorithmic}[1]
		\REQUIRE Image $\mathbf{I}$ of size $M \times N$, grid size $b$, privacy budget $\epsilon$, max variation $m$, simple region mask $\mathbf{M}$ of size $M \times N$, subgrid division factor $n$
		\ENSURE Restored image $\mathbf{I'}$
		\STATE Calculate grid dimensions $G_R \leftarrow \left\lceil \frac{M}{b} \right\rceil$, $G_C \leftarrow \left\lceil \frac{N}{b} \right\rceil$
		\STATE Calculate padding size $P_R \leftarrow G_R \times b - M$, $P_C \leftarrow G_C \times b - N$
		\STATE Pad the image $\mathbf{I}$ and mask $\mathbf{M}$ to size $(M+P_R,N+P_C)$ by mirroring the border pixels
		\STATE Reshape and transpose the padded image $\mathbf{I}_{\text{padded}}$ and mask $\mathbf{M}_{\text{padded}}$ into grids of size $b \times b$: \\
		$\mathbf{I}_{\text{grid}} \leftarrow \text{Reshape}(\mathbf{I}_{\text{padded}}, (G_R, b, G_C, b))$\\
		$\mathbf{I}_{\text{grid}} \leftarrow \text{Transpose}(\mathbf{I}_{\text{grid}}, (0, 2, 1, 3))$\\
		$\mathbf{M}_{\text{grid}} \leftarrow \text{Reshape}(\mathbf{M}_{\text{padded}}, (G_R, b, G_C, b))$\\
		$\mathbf{M}_{\text{grid}} \leftarrow \text{Transpose}(\mathbf{M}_{\text{grid}}, (0, 2, 1, 3))$
		\STATE Compute mask means: $\mu_{\text{mask}} \leftarrow \text{mean}(\mathbf{M}_{\text{grid}},(2, 3))$
		\STATE Identify simple and complex regions:\\
		$\mu_{\text{simple}} \leftarrow \mu_{\text{mask}} > 0.5$\\
		$\mu_{\text{complex}} \leftarrow \mu_{\text{mask}} \leq 0.5$
		\STATE \textbf{Processing in simple regions:}
		\STATE Calculate grid means $\mu_G \leftarrow \text{mean}(\mathbf{I}_{\text{grid}}[\mu_{\text{simple}}], (1,2))$
		\STATE Add noise $\eta \sim \text{Laplace}(0, \sigma)$, where $\sigma = \frac{\Delta}{\epsilon}$, $\Delta = \frac{255 \times m}{b^2}$
		\STATE Clip noisy means $\mu'_G \leftarrow \text{clip}(\mu_G + \eta, 0, 255)$
		\STATE Assign noisy means to simple regions: \\
		$\mathbf{I}_{\text{grid}}[\mu_{\text{simple}}] \leftarrow \mu'_G$, where $\mu'_G$ is expanded to match the grid dimensions.
		\STATE \textbf{Processing in complex regions:}
		\STATE Divide each grid into $n \times n$ subgrids of size $s_b = \frac{b}{n}$ by reshaping the complex regions:\\
		$G_s\leftarrow \text{Reshape}(\mathbf{I}_{\text{grid}}[\mu_{\text{complex}}], (-1, n, s_b, n, s_b))$
		\STATE Calculate subgrid means $\mu_{G_s} \leftarrow \text{mean}(G_s, (2,4))$
		\STATE Set sub variation $s_m \leftarrow m$, sub epsilon $s_\epsilon \leftarrow \epsilon$
		\STATE Add noise $\eta_s \sim \text{Laplace}(0, \sigma_s)$, where $\sigma_s = \frac{\Delta_s}{s_\epsilon}$, $\Delta_s = \frac{255 \times s_m}{s_b^2}$
		\STATE Clip noisy subgrid means:\\
		$\mu'_{G_s} \leftarrow \text{clip}(\mu_{G_s} + \eta_s, 0, 255)$
		\STATE Expand $\mu'_{G_s}$ to full subgrid shape: $G_s \leftarrow \mu'_{G_s}$
		\STATE Assign subgrids back to complex region grids:\\ $\mathbf{I}_{\text{grid}}[\mu_{\text{complex}}] \leftarrow \text{Reshape}(G_s, (-1, b, b))$
		\STATE Reshape and transpose $\mathbf{I}_{\text{grid}}$ back to $\mathbf{I'}$
		\STATE Crop $\mathbf{I'}$ to original size $M \times N$
		\RETURN $\mathbf{I'}$
	\end{algorithmic}
\end{algorithm}

Algorithm~\ref{alg:dp_pixelization_subgrid} implements the proposed differentially private pixelization method with adaptive grid size. It applies differentially private pixelization to an image $\mathbf{I}$ of size $M \times N$ by dividing the image into grids of size $b \times b$ and $s_b \times s_b$. The image $\mathbf{I}$ and simple region mask $\mathbf{M}$ are padded to dimensions $M + P_R$ by $N + P_C$ to ensure that they are divisible by $b$. The padded image $\mathbf{I}_{\text{padded}}$ and mask $\mathbf{M}_{\text{padded}}$ are reshaped into grids $\mathbf{I}_{\text{grid}}$ and $\mathbf{M}_{\text{grid}}$. The mask means $\mu_{\text{mask}}$ are computed to identify simple regions $\mu_{\text{simple}}$ and complex regions $\mu_{\text{complex}}$. For simple regions, grid means $\mu_G$ are calculated, Laplacian noise $\eta$ is added, and the noisy means $\mu'_G$ are clipped and assigned back to $\mathbf{I}_{\text{grid}}$. For complex regions, each grid is divided into $n \times n$ subgrids, where subgrid means $\mu_{G_s}$ are computed, noise $\eta_s$ is added, and the noisy subgrid means $\mu'_{G_s}$ are clipped and assigned back to the main grid $\mathbf{I}_{\text{grid}}$. The processed $\mathbf{I}_{\text{grid}}$ is then reshaped and transposed to reconstruct $\mathbf{I'}$, which is cropped to the original dimensions $M \times N$, resulting in a differentially private pixelated image.
\subsection{Time and Space Complexity of Algorithm~\ref{alg:dp_pixelization_subgrid}} 
\subsubsection{Time Complexity}
\begin{itemize}
\item{Global operations (all regions):}
Step 1–3 (Grid dimensions, padding, mirroring) $\rightarrow$ $O(MN)$;
Step 4–6 (Reshaping, transposing, and mask means) $\rightarrow$ $O(MN)$;
Step 20-21 (Reshaping, transposing, and cropping) $\rightarrow$ $O(MN)$.
\item{Simple regions:}
Step 8 (Calculating grid means) $\rightarrow$ $O(MN)$;
Step 9-10 (Noise and clipping) $\rightarrow$ $O(MN / b^2)$;
Step 11 (Assigning noisy means) $\rightarrow$ $O(MN)$;
\item{Complex regions:}
Step 13 (Dividing into subgrids) $\rightarrow$ $O(1)$;
Step 14 (Calculating subgrid means) $\rightarrow$ $O(MN)$;
Step 15-17 (Noise and clipping) $\rightarrow$ $O(MN \cdot n^2 / b^2)$;
Step 18-19 (Assigning noisy means) $\rightarrow$ $O(MN)$.
\end{itemize}

Total time complexity is linear:

$O\left(MN \left[1 + 1 / b^2 + n^2 / b^2 \right] \right)\Rightarrow O\left(MN \left[1 + n^2 / b^2 \right] \right)$

\subsubsection{Space Complexity}
\begin{itemize}
\item{Global storage:}
Input, padded image/mask $\rightarrow$ $O(MN)$;
Grid representations and output $\rightarrow$ $O(MN)$;
\item{Simple regions:}
Noisy grid means $\rightarrow$ $O(MN / b^2)$;
\item{Complex regions:}
Subgrids $\rightarrow$ $O(MN)$;
Noisy subgrid means $\rightarrow$ $O(MN \cdot n^2 / b^2)$;
\end{itemize}

Total space complexity is linear:

$O\left(MN \left[1 + 1 / b^2 + n^2 / b^2 \right] \right)\Rightarrow O\left(MN \left[1 + n^2 / b^2 \right] \right)$

\subsection{Efficient Storage of the Pixelated Image by structured representation}
\label{subsec:Efficient_Storage_subgrids}
For downstream tasks, we propose storing only the noisy means ($\mu'_G$, $\mu'_{G_s}$), the region mask means ($\mu_{\text{mask}}$), the grid sizes ($b$, $s_b$), and the original image size ($M \times N$), as summarized in Table~\ref{tab:stored_quantities_2}. These metadata are sufficient to reconstruct $\mathbf{I'}$. 
\begin{table}[ht]
	\centering
	\caption{Essential Stored Metadata in Algorithm~\ref{alg:dp_pixelization_subgrid} for Image Reconstruction}
	\begin{tabular}{|c|p{3.2cm}|p{3.2cm}|}
		\hline
		\textbf{Item} & \textbf{Description} & \textbf{Storage Space} \\
		\hline
		$\mu_{\text{mask}}$ & Mask means for each $b \times b$ grid, used to distinguish between simple and complex regions & Proportional to the number of grids, which is $O(MN / b^2)$ \\
		\hline
		$\mu'_G$ & Noisy means of each $b \times b$ grid for simple regions & proportional to the number of grids, which is $O(MN / b^2)$ \\
		\hline
		$\mu'_{G_s}$ & Noisy means of each $s_b \times s_b$ subgrid in complex regions & proportional to the number of subgrids, which is $O(MN \cdot n^2 / b^2)$ \\
		\hline
		$b$, $s_b$ & Simple and complex grid sizes & Constant, i.e., $O(1)$ \\
		\hline
		$M \times N$ & Original image size, required for padding and cropping & Constant, i.e., $O(1)$ \\
		\hline
	\end{tabular}
	\label{tab:stored_quantities_2}
\end{table}

Therefore, the resulting space complexity is given by:
$O(MN / b^2) + O(MN / b^2) + O(MN \cdot n^2 / b^2) + O(1)
\Rightarrow
O\left(MN \left[1 / b^2 + n^2 / b^2 \right]\right)
$. When $n$ is small (e.g., $n = 1$), complex grids are not heavily subdivided, and the overall space reduces to $O(MN / b^2)$. As $n$ increases, the term $O(MN \cdot n^2 / b^2)$ dominates. In the worst case, $n = b$, and the second term becomes $O(MN)$, resulting in a total space complexity of $O(MN)$, which is equivalent to storing the entire image. In summary, compared to storing the full image $\mathbf{I'}$ with $O(MN)$ space, this method achieves significant savings when $n$ is small, but the space cost grows as $n$ increases, reaching $O(MN)$ when $n$ approaches $b$.

For both Algorithm~\ref{alg:parallel_dp_image_pixelization} and Algorithm~\ref{alg:dp_pixelization_subgrid}, restoring the pixelated image $\mathbf{I'}$ from the stored metadata is computationally lightweight. Each $b \times b$ or $s_b \times s_b$ grid is reconstructed by broadcasting a single scalar (the noisy mean) using efficient array expansion without iteration. Algorithm~\ref{alg:parallel_dp_image_pixelization} tiles $\mu'_G$ and crops to the original size, while Algorithm~\ref{alg:dp_pixelization_subgrid} expands $\mu'_G$ and $\mu'_{G_s}$, reshapes and merges them into grids, and then crops the result. These operations rely solely on reshaping and broadcasting, which can be highly optimized on GPUs.

Our storage optimization is useful in both centralized (e.g., cloud storage) and distributed (e.g., edge devices) settings. Compact representations reduce communication overhead and long-term cost in large-scale deployments, while on edge platforms they lower memory usage and bandwidth demand during buffering or streaming. Moreover, the mechanism supports reversibility, as DP-protected images can be reconstructed from compact statistics. This ensures consistent downstream usability while coupling privacy with deployment efficiency, which conventional anonymization cannot achieve.

\section{Experiments and Results}
\subsection{Evaluating the Effectiveness and Efficiency of Parallel Differentially Private Image Pixelization}
\label{Evaluating_Parallel_DP}

\subsubsection{Experimental Setup and Evaluation Metrics}
To evaluate the effectiveness and efficiency of Algorithm~\ref{alg:parallel_dp_image_pixelization}, we conducted experiments on two grayscale-converted datasets: PETS09-S2L1 (795 frames, 768×576, university campus scenes) and Venice-2 (600 frames, 1920×1080, open square scenes). These datasets differ in resolution and context, providing diverse test conditions.
We focus on three key parameters: grid size $b$, maximum pixel variation $m$, and privacy budget $\epsilon$. Effectiveness is measured by mean squared error (MSE) and structural similarity index (SSIM). Using these metrics for evaluation is inspired by \cite{10.1007/978-3-319-95729-6_10}. Efficiency is evaluated using average and total processing time per dataset, excluding metric computation overhead. Experiments were conducted on an NVIDIA RTX A5500 Laptop GPU (16GB VRAM, 7424 CUDA cores) using Python with \texttt{CuPy}, \texttt{NumPy}, and \texttt{scikit-image}.

\begin{figure}[H]
	\centering
	\includegraphics[width=\columnwidth]{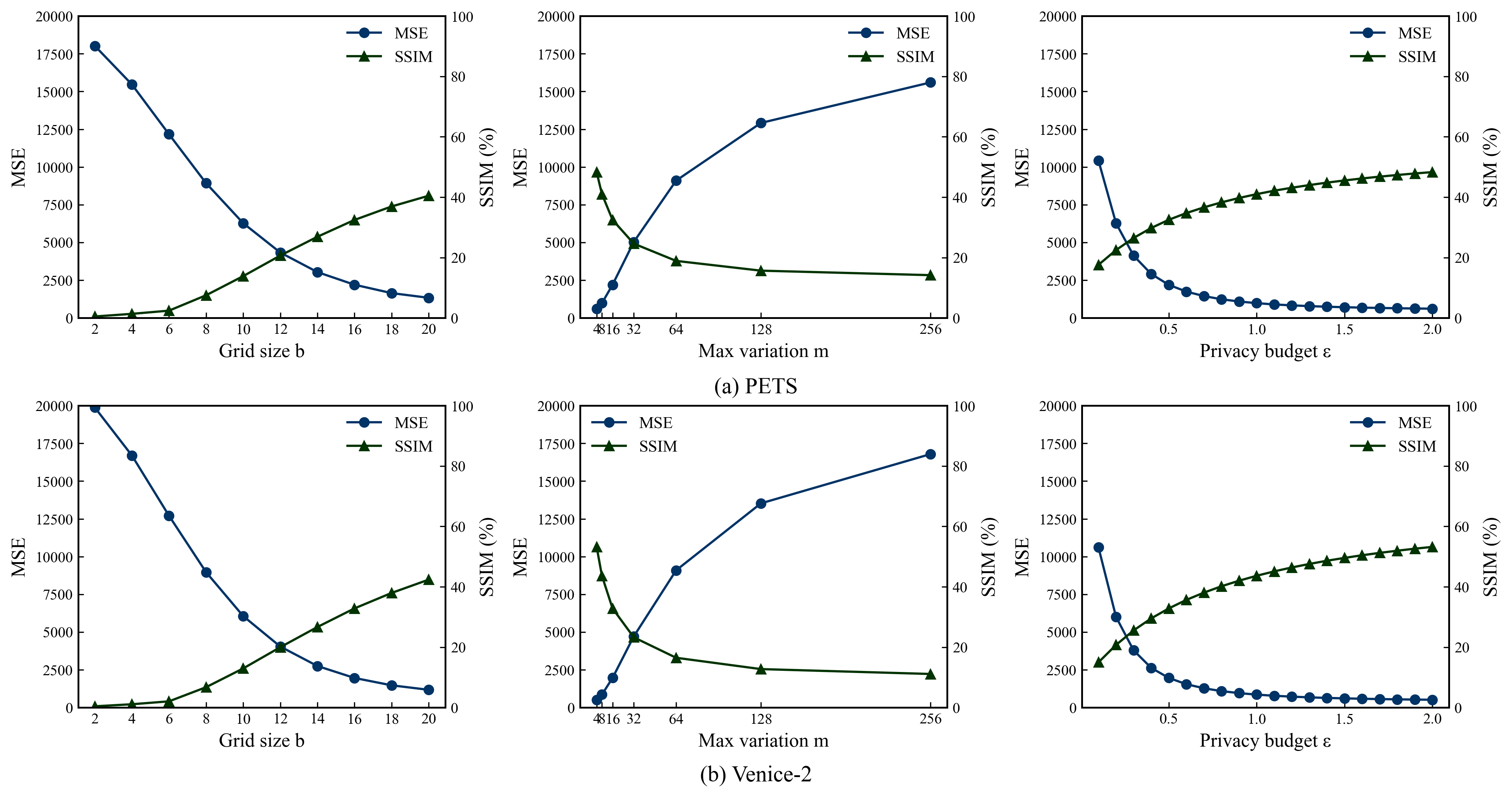}
	\caption{\centering Impact of \( b \), \( m \), and \( \epsilon \) on PETS and Venice-2 datasets: (a) PETS and (b) Venice-2.}
	\label{fig:PETS-Venice-2}
\end{figure}
\subsubsection{Effect of Grid Size \( b \), Maximum Pixel Variation \( m \) and Privacy Budget \( \epsilon \)}
We conducted three experiments: (1) fixing $\epsilon = 0.5$, $m = 16$ and varying grid size $b$ from 2 to 20 to assess utility and efficiency; (2) fixing $b = 16$, $\epsilon = 0.5$ and varying $m$ from 4 to 256 to examine noise effects; and (3) fixing $b = 16$, $m = 16$ while varying $\epsilon$ from 0.1 to 2.0 to evaluate privacy and output quality. Results are shown in Figure~\ref{fig:PETS-Venice-2}. As grid size $b$ increases, fewer partitions lead to lower MSE and higher SSIM, as varying $b$ changes visual granularity and noise per grid. Increasing maximum pixel variation $m$ amplifies noise, degrading SSIM and raising MSE, highlighting the privacy–utility trade-off. A higher privacy budget $\epsilon$ reduces noise, enhancing fidelity while weakening privacy protection.

\begin{figure}[!t]
	\centering
	\includegraphics[width=\columnwidth]{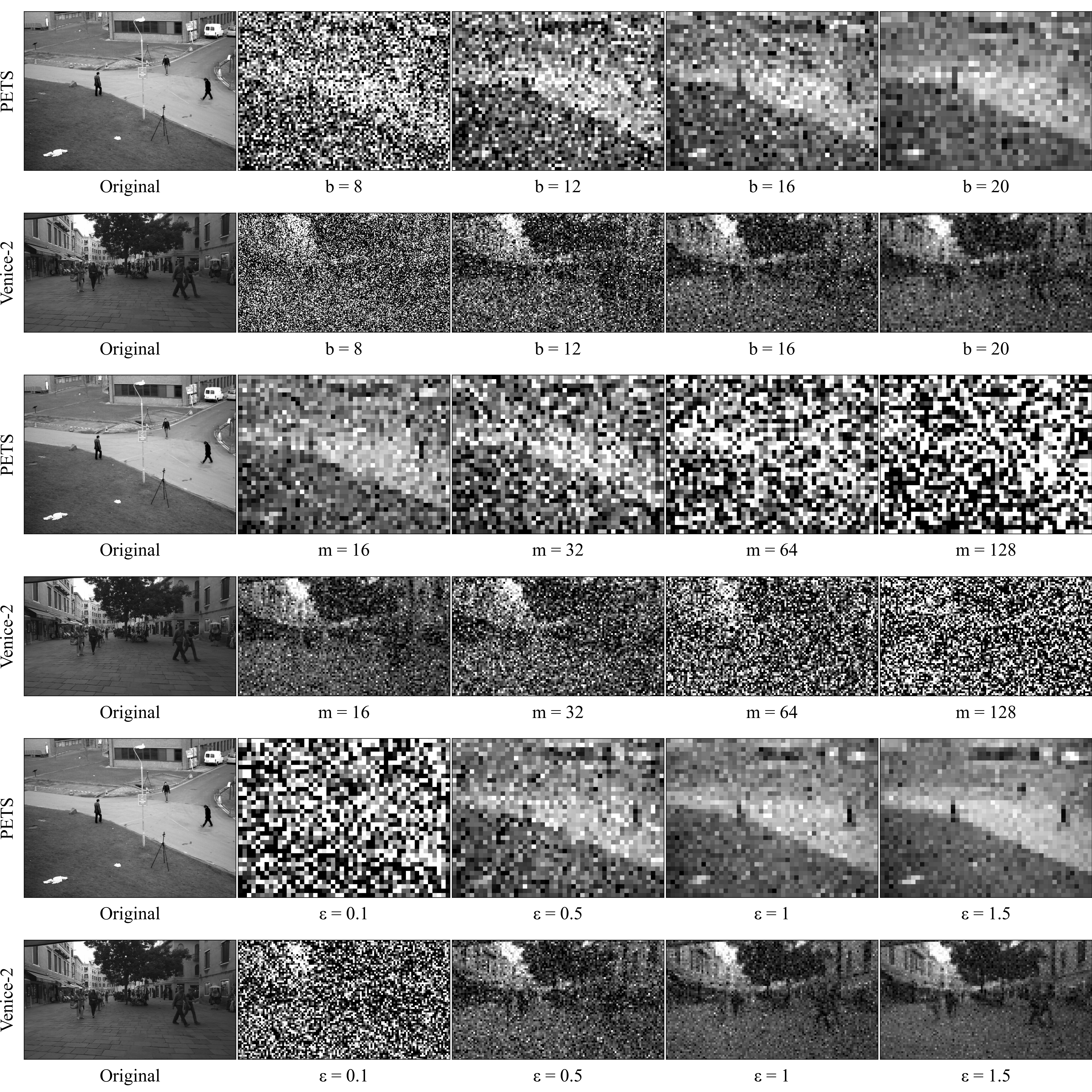}
	\caption{Impact of $b$, $m$, and $\epsilon$ on pixelization performance. Top: $m=16$, $\epsilon=0.5$, $b \in \{8,12,16,20\}$; middle: $b=16$, $\epsilon=0.5$, $m \in \{16,32,64,128\}$; bottom: $b=16$, $m=16$, $\epsilon \in \{0.1,0.5,1,1.5\}$.}
	\label{fig:dp_pixelizationshow}
\end{figure}

Figure~\ref{fig:dp_pixelizationshow} illustrates how $b$, $m$, and $\epsilon$ affect pixelization performance. Varying $b$ changes visual granularity and noise per grid (first two rows), but does not alter the overall privacy guarantee, which is determined solely by $\epsilon$ and $m$. Increasing $m$ adds more noise and reduces fidelity (the middle rows). A higher $\epsilon$ lowers the noise scale, enhancing visual quality while relaxing privacy protection (the last two rows). The results highlight the trade-off between privacy and image quality under various configurations.

\begin{figure}[!t]
	\centering
	\includegraphics[width=0.48\textwidth]{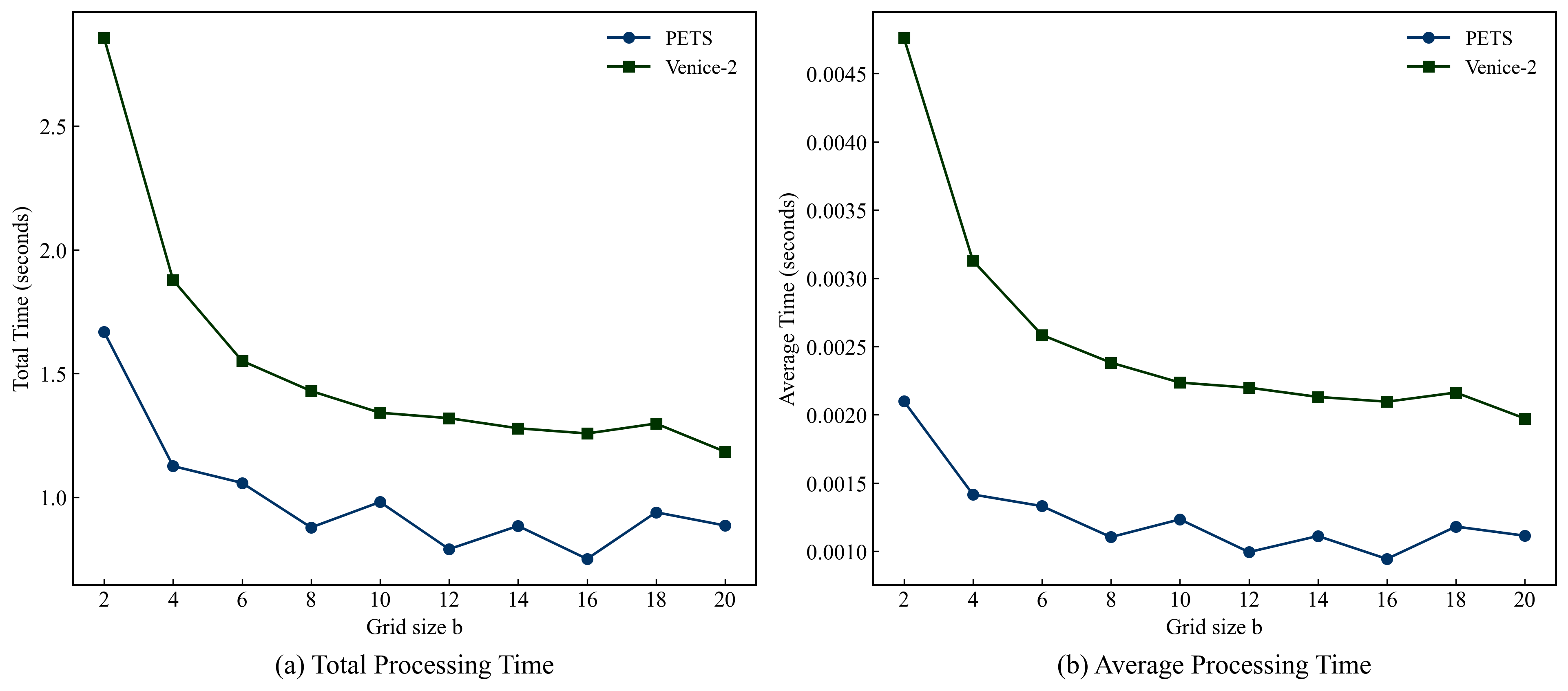}
	\caption{Processing time of Algorithm~\ref{alg:parallel_dp_image_pixelization} versus grid size $b$ on the PETS and Venice-2 datasets. (a) Total processing time. (b) Average processing time.}
	\label{fig:processing_time}
\end{figure}
Fig.~\ref{fig:processing_time} shows how grid size $b$ affects total and average processing time on PETS and Venice-2. In general, larger $b$ reduces computation due to fewer grid operations. However, non-monotonic fluctuations occur when image dimensions are not divisible by $b$, leading to padding overhead (e.g., $b=10,14$ for PETS $768\times576$). For stable runtime and better GPU utilization, we recommend using grid sizes that evenly divide image dimensions, e.g., $b=8,12,16,24$ for PETS and $b=12,24,30,40$ for Venice-2.

To assess runtime efficiency, we compare our method with the classical CPU-based DP pixelization by Fan et al.~\cite{10.1007/978-3-319-95729-6_10}. Their experiments on low-resolution datasets (e.g., AT\&T 92$\times$112, MNIST 28$\times$28) limit applicability to high-resolution scenarios, making them unsuitable for evaluating our GPU-accelerated method. Using fixed $\epsilon=0.5$, $m=16$, and $b=16$, our method achieves 13.3$\times$ and 31.5$\times$ average time speedups on PETS and Venice-2, respectively (Table~\ref{tab:comparison}), reducing total runtime to under 2 seconds. 
\begin{table}[htbp]
	\centering
	\caption{Processing Efficiency Comparison Between Our Method and Fan et al.~\cite{10.1007/978-3-319-95729-6_10}}
	\label{tab:comparison}
	\resizebox{\linewidth}{!}{
		\begin{tabular}{l|c|c|c|c}
			\hline
			Dataset & Method & Avg Time (s) & Total Time (s) & Speedup (Avg Time) \\
			\hline
			PETS & Fan et al. (CPU) & 0.012  & $\sim$9.54  & 1.0$\times$ \\
			PETS & Ours (GPU)       & 0.0009 & 0.7463      & \textbf{13.3$\times$} \\
			Venice-2 & Fan et al. (CPU) & 0.0661  & $\sim$39.66  & 1.0$\times$ \\
			Venice-2 & Ours (GPU)       & 0.0021 & 1.2577      & \textbf{31.5$\times$} \\
			\hline
		\end{tabular}
	}
\end{table}

Our parallel DP pixelization algorithm significantly outperforms the CPU baseline in efficiency. GPU acceleration reduces runtime by over an order of magnitude, enabling real-time processing of high-resolution data. Experiments on PETS and Venice-2 confirm its robustness and scalability.

\subsubsection{Storage Efficiency Evaluation Under Varying Grid Sizes}
\label{Storage_Evaluation_alg2}
To evaluate the practical storage efficiency of our proposed structured representation in Section \ref{Efficient_Storage_Alg_2}, we conducted a full dataset experiment on PETS and Venice-2. For each dataset, we applied Algorithm~\ref{alg:parallel_dp_image_pixelization} using various grid sizes \( b \in \{1, 2, 4, 8, 16, 32, 64, 128\} \). The outputs were saved using two formats in two folders: 
\begin{itemize}
	\item PNG images (\texttt{.png}): Full images $\mathbf{I'}$ saved in PNG format after DP processing.
	\item NPZ files (\texttt{.npz}): Only the noisy means $\mu'_G$ (as \texttt{uint8}) and other metadata (\( b, M, N \)) are stored. The NPZ format (NumPy Zip archive) is a compressed container file that stores multiple NumPy arrays using ZIP compression.
\end{itemize}

Figure~\ref{fig:storage_comparison} presents the total folder sizes for each configuration. As expected, increasing the grid size \( b \) reduces the number of grids and consequently lowers the storage requirement in both formats. However, the reduction is significantly more pronounced for the NPZ format, especially for high-resolution datasets such as Venice-2. When $b=1$, the NPZ format stores one noisy mean per pixel, making compression ineffective. In contrast, PNG benefits from entropy-based compression and can still reduce file size even for pixel-wise outputs, resulting in a smaller overall file. However, at \( b = 4 \), the NPZ-based storage is already about half the size of its PNG counterpart. This gap widens as \( b \) increases; for instance, at \( b = 128 \), the NPZ folder for Venice-2 is only 0.50~MB, whereas the corresponding PNG folder is 2.64~MB—more than five times larger. 
\begin{figure}[H]
	\centering
	\includegraphics[width=\linewidth]{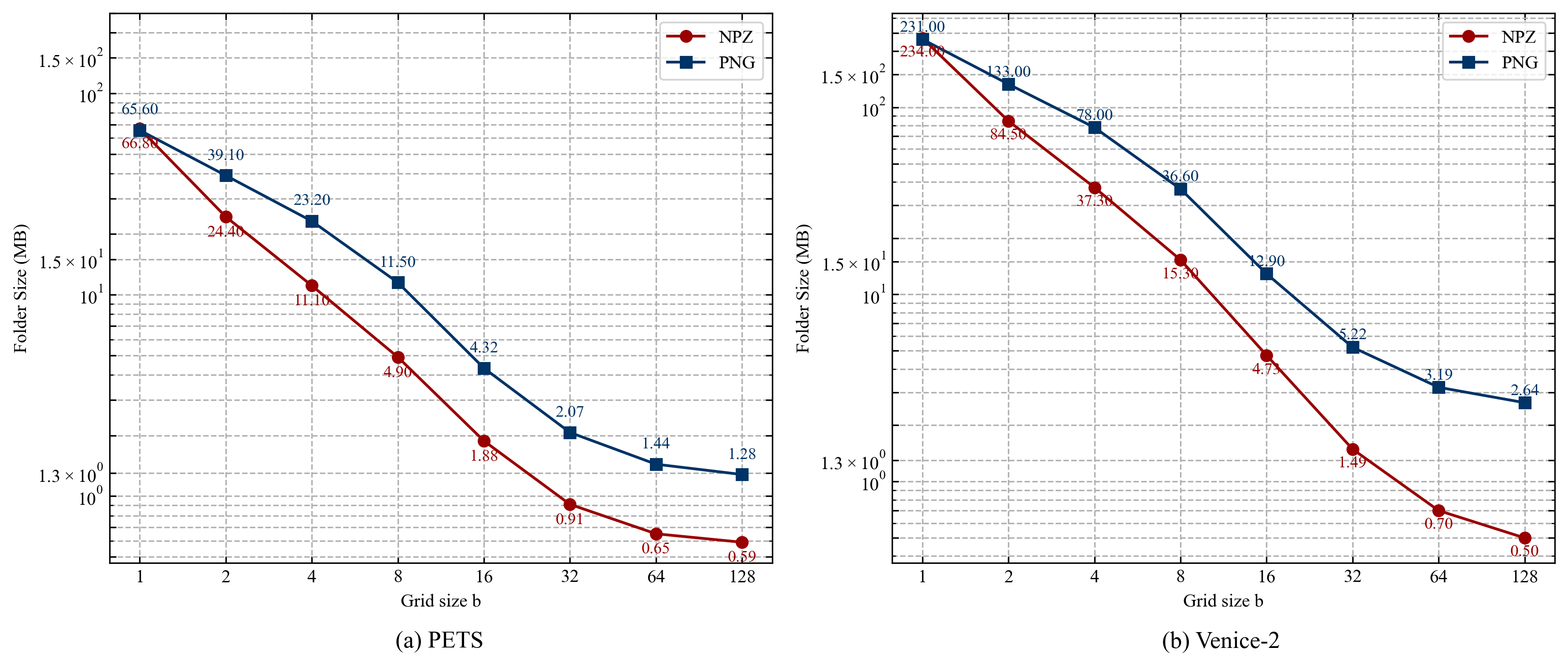}  
	\caption{Storage size comparison of NPZ files and PNG images across different grid sizes $b$. 
		(a) Results on the PETS dataset; (b) Results on the Venice-2 dataset.}
	\label{fig:storage_comparison}
\end{figure}
Although the size of \( \mu'_G \) is approximately reduced to \( 1 / b^2 \) of the full image \( \mathbf{I'} \), the size of the resulting NPZ files does not decrease proportionally compared to PNG images. This is because the PNG format also benefits from powerful entropy-based compression. The highly structured nature of differentially private images makes them particularly compressible in PNG, thereby narrowing the expected storage gap.

These results demonstrate that while PNG benefits from entropy-based compression, it still retains pixel-wise redundancy. In contrast, our NPZ storage captures only the essential anonymized statistics, yielding much better scalability for high-resolution images. Importantly, the NPZ format is not merely a compact representation, it is also fully reversible. The DP images $\mathbf{I'}$ can be faithfully reconstructed from the stored NPZ files.

\subsection{Privacy-Utility Evaluation of Region-Adaptive Parallel Differentially Private Pixelization}
\subsubsection{Fine-Grained Evaluation on Portrait Dataset (PPM-100)}
To assess the privacy–utility trade-off and structural preservation of our region-adaptive Algorithm~\ref{alg:dp_pixelization_subgrid}, we conduct experiments on the PPM-100 dataset. Unlike PETS and Venice-2, which are suitable for evaluating runtime and coarse visual fidelity, PPM-100 contains high-resolution portraits with accurate alpha mattes, enabling detailed analysis of structure-aware privacy. This is essential for evaluating region-adaptive processing based on human silhouettes. Such fine-grained analysis requires precise foreground–background separation, which is unavailable in PETS and Venice-2. Pixel-level annotations in PPM-100 allow objective measurement of segmentation and image fidelity using Intersection-over-Union (IoU), Dice coefficient, MSE, and SSIM, offering insights into how well structural information is preserved under differential privacy. This experiment on the PPM-100 dataset evaluates whether Algorithm~\ref{alg:dp_pixelization_subgrid} preserves image usability after applying differential privacy. We hypothesize that successful segmentation of DP images implies preserved semantic structure, which is essential for downstream tasks such as accident monitoring and activity recognition. Results show that our method maintains structural integrity under privacy constraints, ensuring both functionality and protection.

We use the rembg method (\texttt{u2net\_human\_seg}) to segment human foregrounds from both original and DP images, and compare the results with ground-truth mattes to assess structural preservation after pixelization. 

\subsubsection{Interpreting Structural Fidelity Under Differential Privacy Constraints}
We first fix $\epsilon = 0.5$ and $m = 32$, and systematically vary two parameters: the grid size $b \in \{1, 2, 4, 8, 16, 32, 64, 128\}$ and the subgrid division factor $n \in \{1, 2, 4, 8, \dots, b\}$. Each $b \times b$ grid in the complex region is further split into $n \times n$ subgrids.
Effectiveness is evaluated using four metrics: IoU (DP Image) measures the Intersection-over-Union between the rembg-processed DP image and the matte; Dice (DP Image) reflects the structural alignment using the Dice coefficient; MSE and SSIM measure pixel-level distortion and perceptual similarity relative to the original grayscale images.

\begin{figure}[H]
	\centering
	\includegraphics[width=\columnwidth]{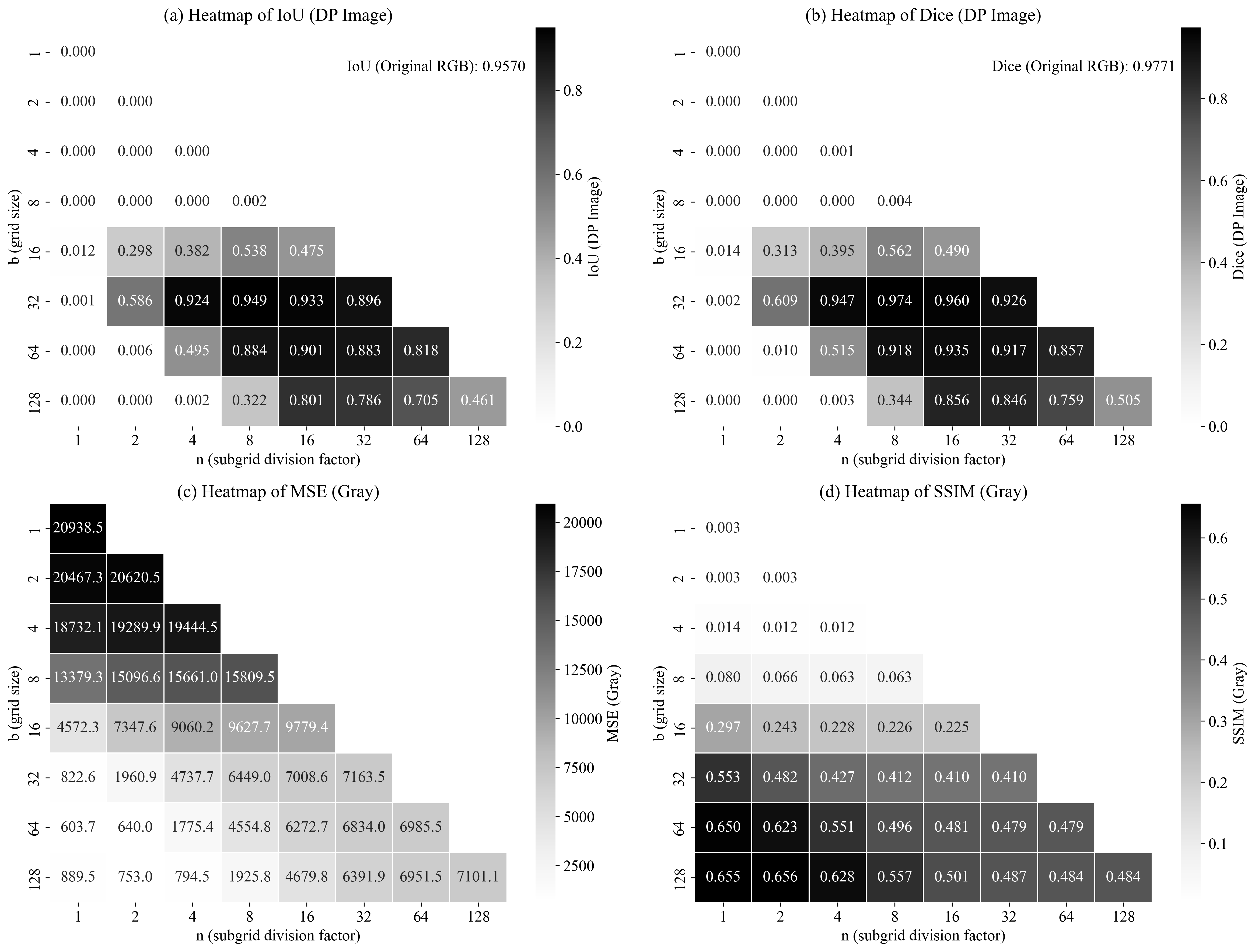}
	\caption{Performance of Algorithm~\ref{alg:dp_pixelization_subgrid}  on PPM-100 dataset across varying \( b \) and \( n \). Metrics include: (a) IoU (DP Image), (b) Dice (DP Image), (c) MSE (Gray), and (d) SSIM (Gray). Original RGB images achieve IoU = 0.9570 and Dice = 0.9771 with the mattes using rembg (indicated in heatmap annotations).}
	\label{fig:heatmap_ppm100}
\end{figure}
The results across all evaluated parameter combinations are visualized in Figure~\ref{fig:heatmap_ppm100}. To provide reference, we also annotate the upper bound performance of original RGB images (without noise) in the heatmaps, which consistently yield IoU = 0.9570 and Dice = 0.9771.
As shown in Figure~\ref{fig:heatmap_ppm100}, when both \( b \) and \( n \) are small (e.g., \( b = 1, 2 \)), the image is excessively fragmented and subjected to strong noise, resulting in high MSE, low SSIM, and nearly zero segmentation performance as IoU and Dice are close to zero. As \( b \) increases to 32 and above, structural fidelity improves significantly. Particularly, for \( b = 32 \), \( n = 8 \), the algorithm achieves IoU = 0.949 and Dice = 0.974, closely matching the original values, while maintaining reasonable distortion (MSE \(= 6449.0 \), SSIM \(= 0.412 \)). This shows that our method can effectively preserve structure while enforcing differential privacy.
However, for overly large \( b \) values (e.g., \( b = 128 \)), although utility metrics such as MSE and SSIM slightly improve due to reduced noise injection per grid, segmentation performance starts to degrade, indicating loss of local structural details crucial to boundary delineation in complex regions.

\textit{Ablation Analysis and Baseline Comparison:}
When $n=1$, the algorithm treats all regions uniformly, equivalent to Algorithm~\ref{alg:parallel_dp_image_pixelization} from the original DP method~\cite{10.1007/978-3-319-95729-6_10}, thus disables region adaptivity and serves as a baseline for comparison. This baseline yields significantly lower IoU and Dice scores (close to 0), highlighting that coarse-grained pixelization alone fails to preserve structural details. In contrast, our region-adaptive method leverages differing granularities for foreground and background regions, resulting in better structural fidelity under DP constraints. The privacy–utility trade-off is further tunable via $b$ and $n$.

\subsubsection{Impact of Privacy Budget $\epsilon$}
We then systematically evaluated the impact of the privacy budget $\epsilon$ on visual fidelity and segmentation performance using the PPM-100 dataset. Parameters were fixed as $b=32$, $n=8$, and $m=32$, while $\epsilon$ varied from $0.1$ to $10.0$. As shown in Figure~\ref{fig:epsilon_analysis}, increasing $\epsilon$ weakens Laplace noise, reducing MSE and increasing SSIM. IoU and Dice improve sharply from $\epsilon=0.1$ to $0.4$, then stabilize and slightly decline beyond $\epsilon=2.0$.
\begin{figure}[H]
	\centering
	\includegraphics[width=\linewidth]{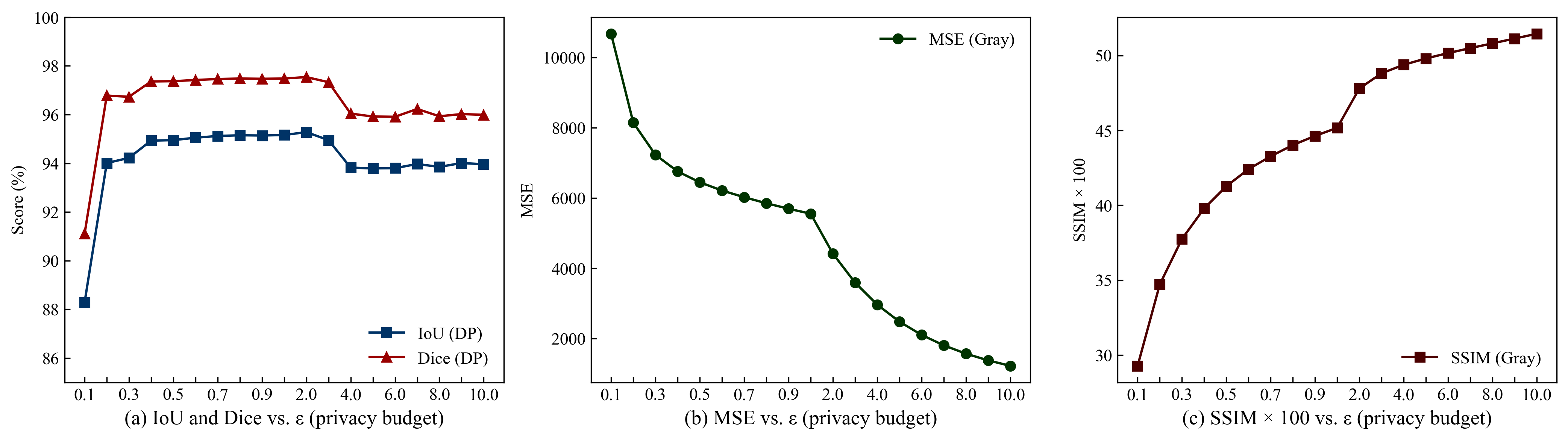}
	\caption{
		Performance of adaptive pixelization on PPM-100 dataset across varying privacy budgets $\epsilon$ with fixed settings ($b=32$, $n=8$, $m=32$).
		(a) IoU and Dice vs. $\epsilon$;
		(b) MSE vs. $\epsilon$;
		(c) SSIM $\times$ 100 vs. $\epsilon$.
	}
	\label{fig:epsilon_analysis}
\end{figure}

This counterintuitive result stems from the model’s sensitivity to distribution shifts. The rembg model (e.g., \texttt{u2net\_human\_seg}) is trained on clean RGB images, whereas pixelization introduces artificial grid artifacts and discontinuities that lead to misclassification along grid edges. Moderate noise levels (e.g., $\epsilon = 1$) help smooth out these artifacts, acting as a regularizer and guiding the model toward meaningful structures. As a result, segmentation peaks at intermediate $\epsilon$ and declines when noise is too high or too low. In privacy-preserving vision, noise serves a dual role: protecting privacy and reducing structural distortion. Properly regularized noise is key to maintaining segmentation performance. For optimal privacy-utility trade-offs, in this experiment, we recommend setting $\epsilon$ within the range $[0.5, 1.0]$. In this region, the segmentation performance is nearly identical to the original images, while privacy guarantees remain meaningful.

\begin{figure}[!t]
	\centering
	\includegraphics[width=\columnwidth]{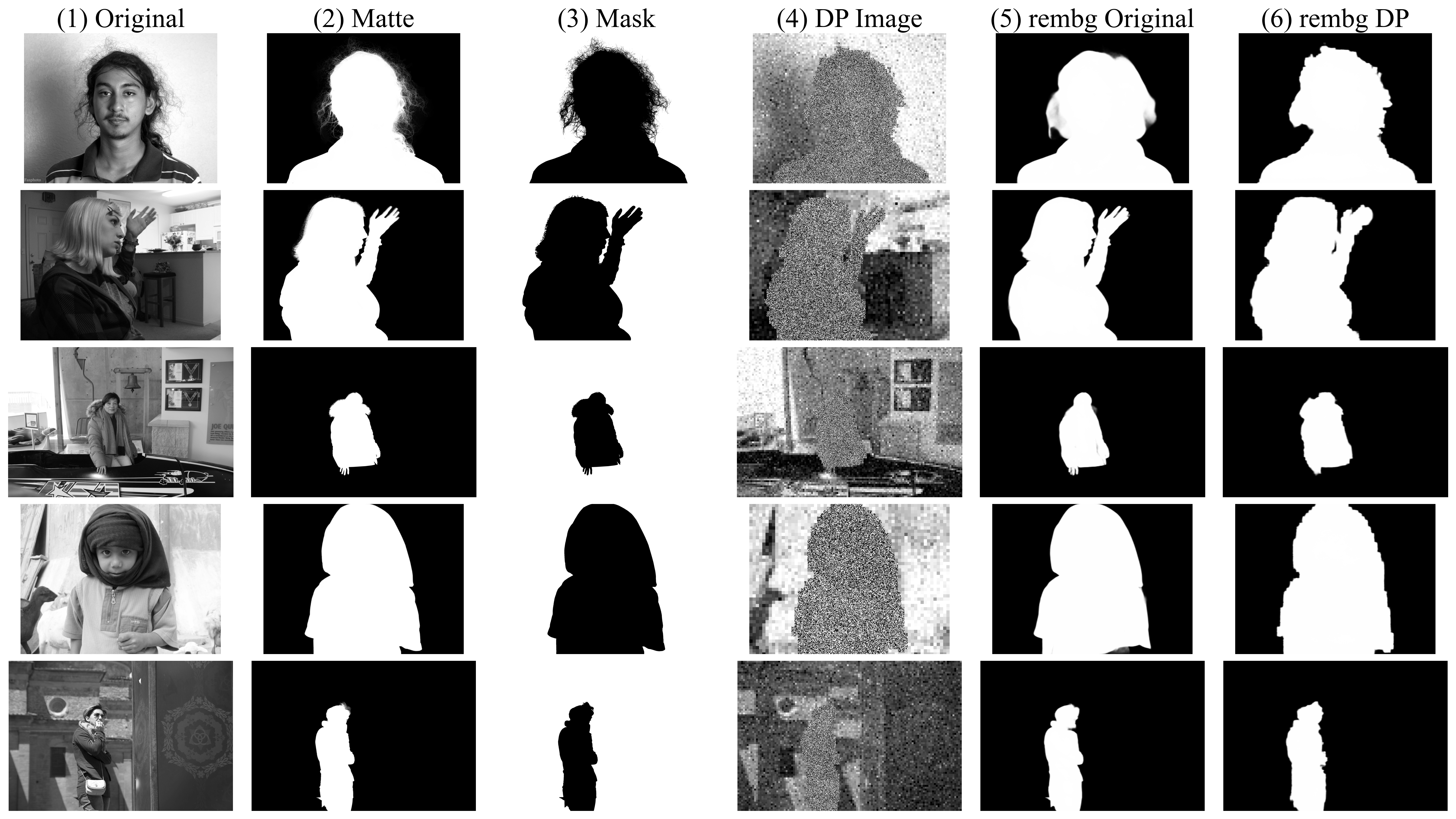}
	\caption{\centering Visualization of region-adaptive differentially private pixelization. Each row corresponds to one image from PPM-100, showing: (1) original grayscale image, (2) alpha matte, (3) binary region mask, (4) pixelized DP image, (5) extracted mask from original image, and (6) extracted mask from DP image.}
	\label{fig:region_adaptive_visualization}
\end{figure}
\subsubsection{Visual Analysis of Structural Preservation Under Differential Privacy}
To intuitively illustrate the effectiveness of our region-adaptive differentially private pixelization algorithm, we conduct a visual analysis on five high-resolution portrait images randomly selected from the PPM-100 dataset. As shown in Figure~\ref{fig:region_adaptive_visualization}, the algorithm divides the image into simple and complex regions using a binary mask derived from the alpha matte: simple background regions are pixelized with coarse grids, while complex foreground regions are refined with subgrid noise injection. The experiment is conducted with fixed parameters: $b=32$, $m=32$, $\epsilon=0.5$, and $n=8$. All operations are accelerated using CuPy on the GPU. For each image, we show six panels: (1) original grayscale image, (2) ground-truth alpha matte, (3) binary mask used for region division, (4) DP pixelized result, (5) foreground mask extracted from the original image using rembg, and (6) mask extracted from the DP image. Comparing the rembg masks from original and DP images reveals how well human contours are preserved after pixelization. A good match in the DP mask indicates that key structural features remain intact despite added noise. Although fine details may degrade, the region-adaptive method retains overall shape. This ensures the DP image remains usable for tasks like detection or segmentation, effectively balancing privacy protection and image utility.

\subsubsection{Storage Efficiency Evaluation Under Varying Subgrid Division Factors}
While the impact of the base grid size \( b \) on storage efficiency has already been analyzed in Section~\ref{Storage_Evaluation_alg2}, we now focus on evaluating the effect of varying the subgrid division factor \( n \). Larger values of \( n \) imply finer granularity in complex regions, resulting in more subgrids and greater storage demand for the corresponding noisy means. For each \( n \in \{1, 2, 4, 8, 16, 32, 64, 128\} \), we applied Algorithm~\ref{alg:dp_pixelization_subgrid} to the full PPM-100 dataset with fixed $b=128$, $\epsilon = 0.5$, and $m=32$. The outputs were saved using two formats in two folders: 
\begin{itemize}
	\item PNG images (\texttt{.png}): Full images $\mathbf{I'}$ saved in PNG format after DP processing.
	\item NPZ files (\texttt{.npz}): Only the simple noisy means $\mu'_G$ (as \texttt{uint8}), complex noisy means $\mu'_{G_s}$ (as \texttt{uint8}), mask means $\mu_{\text{mask}}$ (as \texttt{float32}), and other metadata (\( b, s_b, M, N \)) are stored. To ensure stable image reconstruction, $\mu_{\text{mask}}$ should be stored as \texttt{float32}. Using \texttt{uint8} causes precision loss, which may lead to errors during recovery.
\end{itemize}

Figure~\ref{fig:n_vs_storage} presents the storage cost versus subgrid division factor \( n \). The size of PNG output increases steadily with \( n \), while the NPZ metadata shows a sharper growth trend. This conforms to the space complexity analysis in Section~\ref{subsec:Efficient_Storage_subgrids}. In practice, the subgrid division factor \( n \) should be selected based on the application's storage-utility trade-off. Nevertheless, in all tested settings, our proposed NPZ-based metadata storage remains smaller than the PNG output. Importantly, this metadata-based storage is reversible: the full differentially private image $\mathbf{I'}$ can be faithfully reconstructed from the metadata. This enables more efficient storage, transmission, and reusability of privacy-preserving visual data.

\begin{figure}[t]
	\centering
	\includegraphics[width=0.37\textwidth]{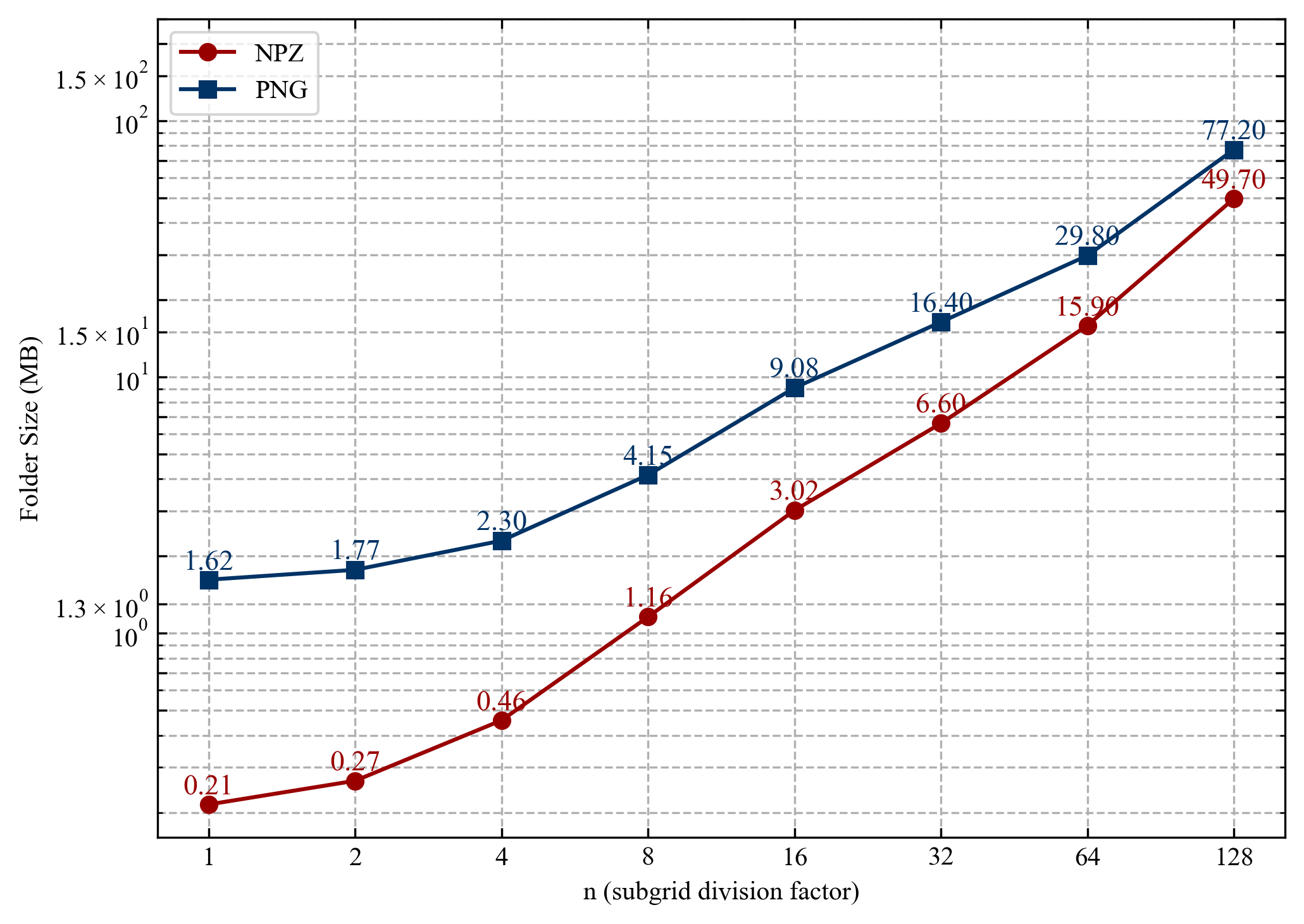}
	\caption{Storage size comparison of NPZ files and PNG images across different subgrid division factor \( n \).}
	\label{fig:n_vs_storage}
\end{figure}

\subsection{Face Re-Identification Attack Evaluation on CelebA via ArcFace}
\subsubsection{Limitations of PPM-100 and Motivation for Using CelebA}
To further evaluate privacy protection, we conducted a face re-identification (Re-ID) experiment using a pre-trained ArcFace model~\cite{8953658} on CelebA dataset. PPM-100 lacks identity labels and is thus used for structure fidelity evaluation, while CelebA provides labeled identities but lacks alpha mattes. Hence, the two datasets serve complementary roles: PPM-100 evaluates visual fidelity; CelebA assesses identity leakage resistance.
ArcFace is chosen for its high accuracy and widespread use in face recognition benchmarks. It generates discriminative embeddings, making it suitable for evaluating identity leakage risk. The evaluation assumes a strong adversary with access to clean gallery images, attempting to identify individuals by comparing embeddings of privacy-protected query images with those of the gallery.

\subsubsection{Experimental Setup}
We selected 100 identities from CelebA, each with two images. For each identity, one image was used as a gallery reference and the other as a query. The gallery images remained unprotected. The query images were first masked using BiSeNet \cite{yu2018bisenet} to distinguish complex (face) and simple (background) regions. We then applied our region-adaptive differentially private pixelization to the queries under various parameter configurations:
\begin{itemize}
	\item Varying grid size \( b \in \{4, 8, 12, 16, 20\} \)($\epsilon=5$, $n=4$)
	\item Subgrid division factor \( n \in \{1, 2, 4, 8, 16\} \)($\epsilon=5$, $b=16$)
	\item Privacy budget \( \epsilon \in \{0.1, 2.5, 5.0, 7.5, 10.0\} \)($b=16$, $n=4$)
\end{itemize}
Fig.~\ref{fig:first_query_image} shows the differentially private result of the first query image as a representative example.
\begin{figure}[H]
	\centering
	\includegraphics[width=\columnwidth]{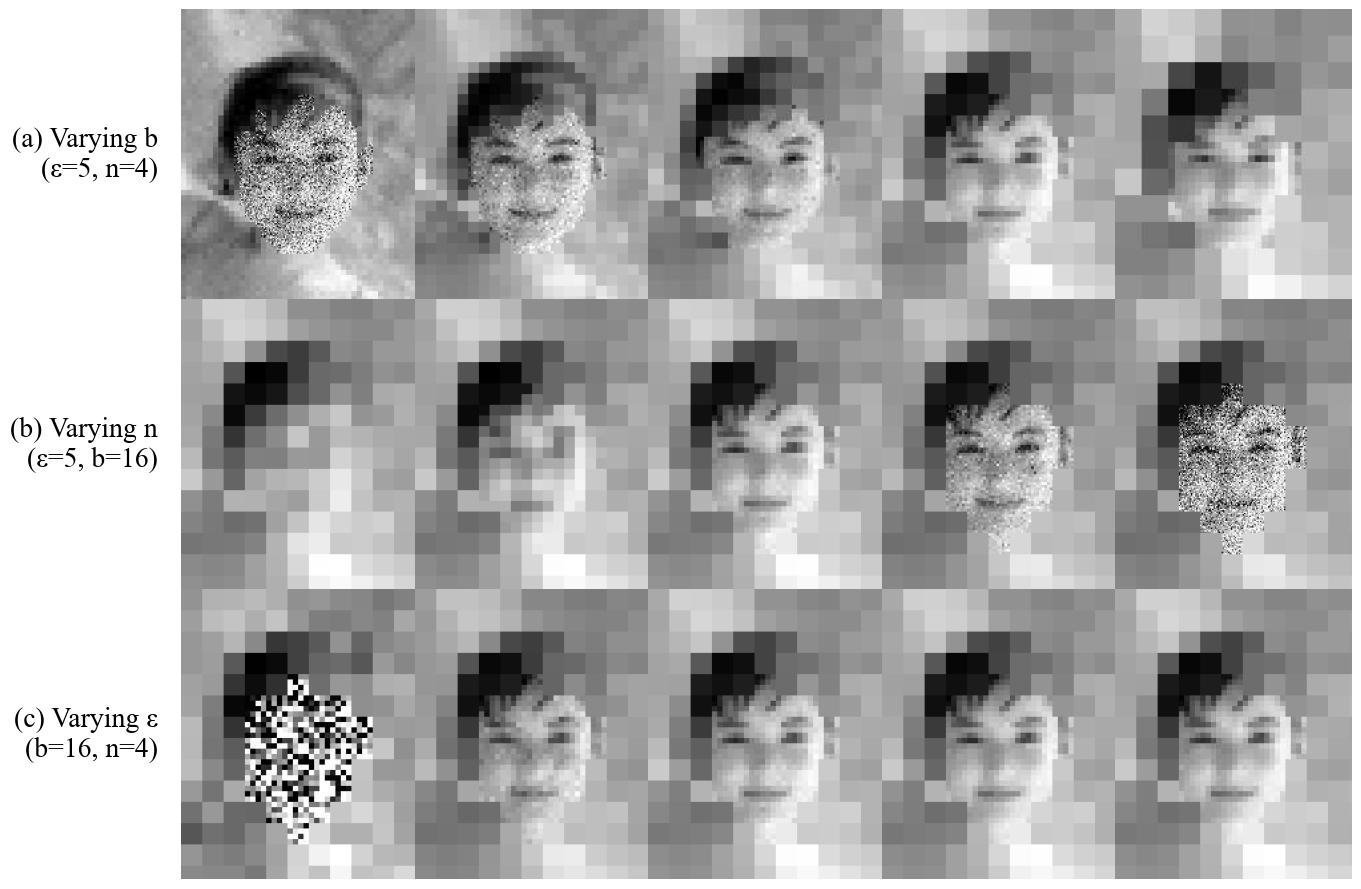}
	\caption{The differentially private result of the first query image under various parameter configurations. (a) Varying grid size \( b \), ($\epsilon=5$, $n=4$); (b) Varying subgrid division factor \( n \), ($\epsilon=5$, $b=16$); (c) Varying privacy budget \( \epsilon \), ($b=16$, $n=4$).}
	\label{fig:first_query_image}
\end{figure}

Features were extracted using ArcFace, and cosine similarity was used to compute Top-1 matching accuracy. To avoid Top-1 accuracy collapse under excessive noise, we selected a large privacy budget range $\epsilon \in \{0.1, 2.5, 5.0, 7.5, 10.0\}$ and a small maximum variation parameter $m=1$, yielding low noise with scale $\frac{255m}{b^2\epsilon}$. This preserves partial recognition and enables meaningful Top-1 accuracy variation across $b$, $n$, and $\epsilon$. In contrast, smaller $\epsilon$ or larger $m$ would overwhelm the image with noise, resulting in uniformly low Top-1 accuracy (close to 1.0\%), which would obscure the influence of pixelization parameters and hinder comparative analysis.

This experiment adopts an attacker's perspective, assuming access to both the gallery and privacy-protected query images to evaluate re-identification risk. Crucially, the attacker has no access to unprotected queries for training, reflecting a realistic threat model where recognition operates solely on perturbed visual content exposed after privacy transformation.

\subsubsection{Results and Analysis}
The results are shown in Fig.~\ref{fig:reid_accuracy}. Each curve illustrates the Top-1 accuracy trend under one parameter setting. We make the following observations:
\begin{itemize}
	\item Grid size $b$: In Fig.~\ref{fig:reid_accuracy} (a), accuracy rises with $b$ up to 12 (57\%), then declines at $b=16$ and $b=20$, showing that overly large grids disrupt local structure. This reveals a non-monotonic utility trend with respect to $b$.
	\item Subgrid division factor $n$: Fig.~\ref{fig:reid_accuracy} (b) shows accuracy peaking at $n=4$ (43\%) but dropping with larger $n$ due to over-fragmentation and noise. It exhibits a clear non-linear effect on recognition accuracy. 
	\item Privacy budget $\epsilon$: As seen in Fig.~\ref{fig:reid_accuracy} (c), accuracy increases with $\epsilon$, from 1\% ($\epsilon=0.1$) to 48\% ($\epsilon=10.0$), suggesting diminishing gains in utility when privacy protection is already weak.
\end{itemize}
\begin{figure}[H]
	\centering
	\includegraphics[width=\linewidth]{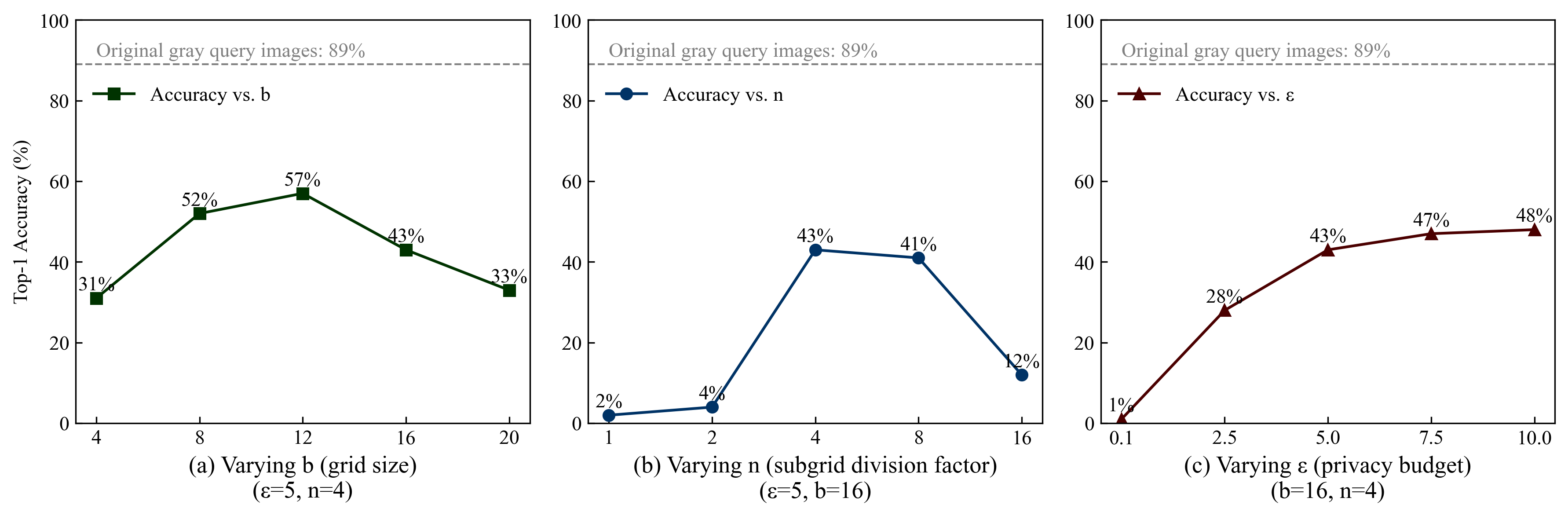}
	\caption{
		Top-1 face re-identification accuracy under varying parameters. 
		(a) $b$ ($\epsilon=5$, $n=4$); 
		(b) $n$ ($\epsilon=5$, $b=16$); 
		(c) $\epsilon$ ($b=16$, $n=4$). 
		Gray dashed line: The accuracy is 89\% from unprotected grayscale queries. Accuracy is computed via ArcFace embeddings and cosine similarity.
	}
	\label{fig:reid_accuracy}
\end{figure}
The accuracy drops significantly under all privatized settings compared with the accuracy on unprotected grayscale queries (89\%), confirming the effectiveness of our region-adaptive pixelization in suppressing re-identification. Furthermore, the flexibility of parameters enables customizable privacy–utility trade-offs, and complements the PPM-100 results by demonstrating robustness against identity inference.

\subsubsection{Ablation Analysis and Baseline Comparison} 
In Fig.~\ref{fig:reid_accuracy} (b), the baseline case $n=1$ applies uniform noise without region differentiation~\cite{10.1007/978-3-319-95729-6_10}, yielding strong privacy but significantly degraded utility (Top-1 accuracy $=$ 2\%). This highlights the cost of undifferentiated protection. In Fig.~\ref{fig:reid_accuracy} (c), increasing $\epsilon$ improves accuracy, indicating weaker privacy. This demonstrates that when no noise is added, privacy protection is insufficient. These trends emphasize the need for balanced parameter tuning to maintain both privacy and utility.

\section{Discussion}
The parallel region-adaptive DP method preserves formal differential privacy, ensuring that adding or removing $m$ pixels yields only bounded changes in the output, governed by $\epsilon$.  While differential privacy provides provable guarantees, it does not ensure full perceptual anonymity or prevent recognition from structural cues such as body contours. We do not claim to prevent all forms of human or model-based recognition.

Our work focuses on improving the efficiency and utility of the DP pixelization process and reducing storage overhead through lightweight representation. While region-adaptive mask generation (e.g., via segmentation or face detection) can be time-consuming, it is typically shared across tasks and performed once per image. This component can be further optimized in future work, whereas our current contribution lies in accelerating the DP pixelization stage and enabling compact, reversible storage for practical deployment.

Our region-adaptive parallel DP framework has potential for application in the following domains: in human pose estimation and fall detection, small grids are applied to joints and body contours to retain motion cues, while coarse grids are used for static backgrounds; in facial analysis, fine grids preserve local features such as eyes, nose, and mouth, while larger grids obscure forehead or cheeks; in video conference analysis, small grids maintain the structure of facial and hand regions, whereas larger grids anonymize the background; in autonomous driving, detailed structures like road signs and pedestrian contours use small grids, while sky and distant objects are coarsely pixelated; and in retail monitoring, small grids capture customer hand-object interactions, while less informative areas like ceilings or aisle spaces use large grids. 

\section{Conclusion}
We propose a region-adaptive pixelization framework that enforces differential privacy with high utility and real-time performance. By assigning finer grids and stronger noise to sensitive regions (e.g., faces), and coarser grids to backgrounds, our method preserves essential visual structure while enhancing privacy. The grid size $b$ and privacy budget $\epsilon$ can be tuned to balance structural fidelity and privacy strength, making the approach adaptable to diverse applications such as surveillance, healthcare, and activity recognition. GPU-based parallelization enables efficient deployment on high-resolution data. A compact storage design further improves efficiency. Experiments on PPM-100 confirm that tuning $b$ and the subgrid division factor $n$ helps retain human contours. On CelebA, our method resists face re-identification even under weak noise, demonstrating robust protection against identity inference.

We acknowledge that our method, while formally satisfying differential privacy, may still retain perceptible structural cues in critical regions. This raises an important distinction between formal privacy guarantees and perceptual anonymity, which we believe deserves further investigation by the research community. By sharing this work, we aim to provide a foundation for exploring hybrid approaches that combine differential privacy with structural obfuscation.

\bibliographystyle{IEEEtran} 
\bibliography{mybibfile}
\end{document}